\documentclass[12pt]{article}
\usepackage{authblk}
\usepackage{geometry}
\usepackage{amsmath}
\usepackage{amssymb}
\usepackage{amsthm}
\usepackage{mathrsfs}
\usepackage{subfigure}
\usepackage{customcommands}
\usepackage{bm}
\usepackage{Nettastyle}

\usepackage{amsmath}
\usepackage{amssymb}
\usepackage{amsfonts}
\usepackage{braket}
\usepackage[normalem]{ulem}
\usepackage{color}
\usepackage{bbold}
\usepackage{ulem}
\usepackage{mathtools}
\usepackage{tensor} 

\title{A World without Pythons would be so Simple}

\author[1]{Netta Engelhardt,}
\author[2, 3]{Geoff Penington,}
\author[4]{and Arvin Shahbazi-Moghaddam}

\affiliation[1]{Center for Theoretical Physics, Massachusetts Institute of Technology, \\Cambridge, MA 02139, USA}
\affiliation[2]{Center for Theoretical Physics and Department of Physics,\\
University of California, Berkeley, CA 94720, U.S.A. and}
\affiliation[3]{Institute for Advanced Study, 1 Einstein Dr, Princeton, NJ 08540, U.S.A.} 
\affiliation[4]{Stanford Institute for Theoretical Physics,\\ Stanford University, Stanford, CA 94305 USA}
\emailAdd{engeln@mit.edu}
\emailAdd{geoffp@berkeley.edu}
\emailAdd{arvinshm@gmail.com}

\abstract{We show that bulk operators lying between the outermost extremal surface and the asymptotic boundary admit a simple boundary reconstruction in the classical limit. This is the converse of the Python's lunch conjecture, which proposes that operators with support between the minimal and outermost (quantum) extremal surfaces -- e.g. the interior Hawking partners -- are highly complex. Our procedure for reconstructing this ``simple wedge'' is based on the HKLL construction, but uses causal bulk propagation of perturbed boundary conditions on Lorentzian timefolds to expand the causal wedge as far as the outermost extremal surface. As a corollary, we establish the Simple Entropy proposal for the holographic dual of the area of a marginally trapped surface as well as a similar holographic dual for the outermost extremal surface. We find that the simple wedge is dual to a particular coarse-grained CFT state, obtained via averaging over all possible Python's lunches. An efficient quantum circuit converts this coarse-grained state into a ``simple state'' that is indistinguishable in finite time from a state with a local modular Hamiltonian. Under certain circumstances, the simple state modular Hamiltonian generates an exactly local flow; we interpret this result as a holographic dual of black hole uniqueness.}
\begin{document}
\maketitle

\section{Introduction} \label{sec:intro}

Recent developments in the black hole information frontier have pointed to a holographic geometrization of the degrees of freedom of the Hawking radiation~\cite{AEMM, Pen19, AlmMah19a, AlmHar19, PenShe19, BroGha19}. For an AdS black hole evaporating into a bath, the ``entanglement wedge of the radiation’’ after the Page time includes a large part of the black hole interior, bounded by the minimal quantum extremal surface (QES)~\cite{EngWal14}. 

This geometric description of the information naturally accounts for both the Page curve~\cite{Pag93b} and the Hayden-Preskill decoding criterion~\cite{HayPre07}. It also leads to a geometrical explanation for the expectation of Harlow-Hayden~\cite{HarHay13} that decoding Hawking radiation should be exponentially complex. Even though the interior degrees of freedom lie on the radiation side of the minimal QES and so lie within the radiation entanglement wedge, they are still hidden behind a nonminimal QES; in the case of the single-sided black hole, the nonminimal QES is simply the empty set. 
The region between the nonminimal and minimal extremal surfaces was dubbed ``the Python's lunch'' in~\cite{BroGha19}, because appropriate Cauchy slices in the bulk (quantum) geometry have a constriction at each extremal surface, together with a bulge in the middle (the eponymous ``lunch''). See Fig.~\ref{fig:lunchintro} for an illustration.
\begin{figure}
    \centering
\label{fig:lunchintro}
    \includegraphics[width=0.9\textwidth]{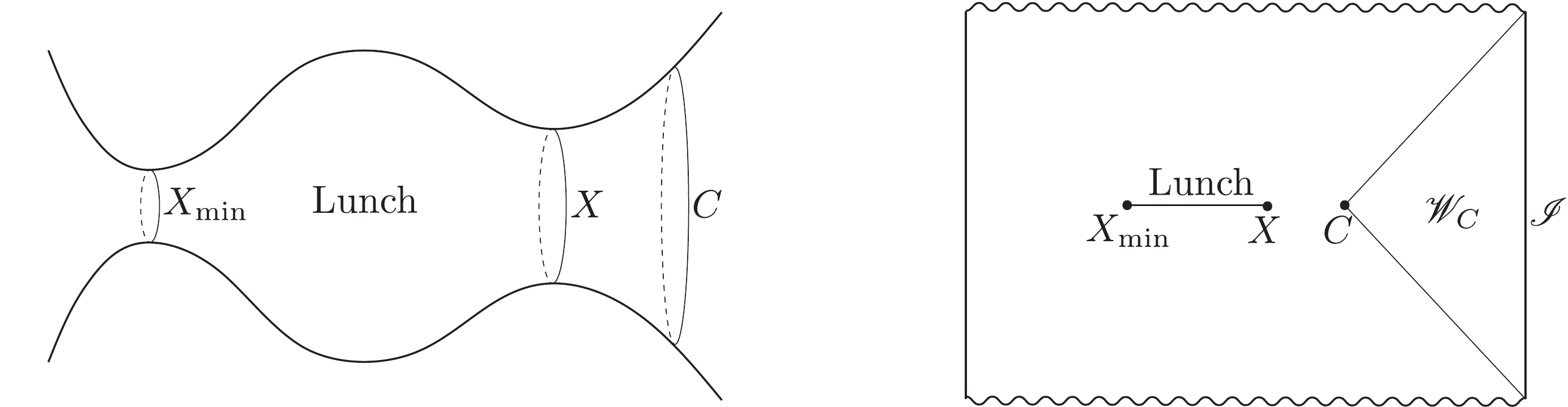}
    \caption{An illustration of the Python's lunch. On the left, the geometry of a Cauchy slice featuring the titular python's lunch between the two constrictions, the dominant QES $X_{\mathrm{min}}$ and the non-minimal QES $X$, both of which lie behind the causal surface $C$. On the right, a spacetime diagram of the same. }
\end{figure}

The claim of~\cite{BroGha19} was that any bulk operator with support in the interior of a Python's lunch should be exponentially difficult to decode, with an exponent that is controlled by the size of the bulge and grows as $O(1/G_N)$ in the semiclassical limit. The justification for this conjecture was based primarily on tensor network toy models, where the fastest known protocols for decoding operators inside a lunch use a Grover-search-based algorithm that takes exponential time. 

An additional important consistency check comes from the quantum focusing conjecture~\cite{BouFis15a}, which is the quantum avatar of classical gravitational lensing when the null energy condition is violated due to quantum corrections. Quantum focusing and global hyperbolicity ensure that no causal semiclassical Lorentzian evolution of the bulk geometry can result in causal communication from behind a Python's lunch to the asymptotic boundary, no matter what the asymptotic boundary conditions are. If such communication were possible, reconstruction of information from within the lunch could be implemented using only boundary time evolution with simple sources and the extrapolate dictionary relating bulk observables at the asymptotic boundary to local boundary operators. From a boundary perspective, this is a very simple procedure compared to the exponential complexity claimed to be necessary for operator reconstruction inside the lunch.

Without input from nonperturbative quantum gravity (such as entanglement wedge reconstruction), simple reconstruction using only low-complexity, causally-propagating operators and sources is all that semiclassical Lorentzian gravity is capable of: i.e. all that semiclassical gravity ``sees’’. Calculations and arguments that rely exclusively on semiclassical gravity with no further input (e.g. Hawking’s original calculation) are thus certainly restricted by the Python’s lunch proposal to recover no more than the domain of dependence between the outermost quantum extremal surface -- the ``appetizer'' of the lunch -- and the asymptotic boundary.

So how much does purely semiclassical gravity \textit{actually} recover? On the one hand, if the Python’s lunch conjecture is true, it is natural to expect that simple reconstruction can in fact obtain the entire bulk up to the outermost extremal surface. This ``converse'' to the Python's lunch conjecture is certainly true in tensor network toy models where anything not in a lunch can be reconstructed using a simple unitary circuit. Any gap in gravitational theories between the simply reconstructible region and the start of the lunch would therefore be somewhat puzzling and demand explanation. On the other hand, simple reconstruction appears to be little more than a glorified version of the HKLL procedure~\cite{HamKab05, HamKab06, HamKab06b}, which is supposed to recover just the so-called causal wedge: the region that can simultaneously send and receive signals from the asymptotic boundary. And generically the causal wedge and the outermost extremal wedge do not coincide.

To clarify this conundrum, let us first briefly review HKLL here, as it will be instrumental for our work in this paper. The HKLL procedure is a reconstruction protocol for bulk matter on a fixed background, in which the bulk fields (which can include gravitons~\cite{Hee12, KabLif12}, $\mathcal{O}(1/N)$ corrections~\cite{KabLif12, HeeMar12}, interactions~\cite{KabLif11}, and higher spins~\cite{KabLif12, SarXia14}) are obtained via a non-standard Cauchy evolution from their boundary counterparts (related to them via the extrapolate dictionary). Quantitatively,
\be \label{eq:HKLL}
\phi(x) =\int dX K(x;X)\mathcal{O}(X),
\ee
where $K(x;X)$ is a smearing function that depends on the spacetime geometry supported on the set of boundary points spacelike-separated to $x$.\footnote{Note that the validity of this non-standard ``rotated'' Cauchy problem is far from well-established (though see~\cite{Holmgren, Tat07} for proofs in certain cases). Our purpose here is not to put HKLL on a firm footing, but rather to show that reconstruction of the simple wedge is as simple as HKLL. We shall therefore assume HKLL, but any other simple reconstruction procedure for operators in the causal wedge would do just as well for our purposes.} The sense in which HKLL is ``simple'' is evident: from a boundary perspective it consists of boundary time evolution with local sources turned on. And local Hamiltonian evolution can be simulated efficiently using a quantum circuit.

The immediate prediction therefore, as expressed in~\cite{BouFre12}, is that HKLL can reconstruct operators within the causal wedge. We might then expect that the simply reconstructible region -- which we shall henceforth refer to as the \textit{simple wedge} -- is to be identified with the causal wedge. Since the causal wedge is always a subset of the outermost quantum extremal wedge~\cite{EngWal14} and is generically a proper subset~\footnote{As shown in \cite{AlmMah19b}, quantum effects can allow the causal wedge to be outside the outermost extremal wedge when defined using time evolution couples the asymptotic boundary to an auxiliary system. However, this is only true if we do not include the auxiliary coupled system when defining the outermost extremal wedge. When comparing apples to apples by doing so, one indeed finds that the causal wedge is still contained in the outermost extremal wedge.}, this leads to the undesirable no-man’s land between the simple wedge and the outermost quantum extremal surface.

To see deeper into the bulk, we need to expand the causal wedge via the addition of simple boundary sources as proposed in~\cite{EngWal17b, EngWal18}. In the very special case when the gap between the outermost extremal surface and the causal wedge is Planckian, \cite{Levine:2020upy} showed that certain causal unitaries produce just enough backreaction to maximally expand the causal past or future. However, in generic spacetimes the gap region is non-empty even in the classical limit and can in fact be arbitrarily large. It was conjectured in~\cite{EngWal17b, EngWal18} that it should be possible to fine-tune simple sources in order to ``turn off'' any extant focusing and so expand the causal wedge up to an apparent horizon, all without violating the null energy condition.

A central result of this paper is an explicit, constructive derivation of this fact -- in the limit where the bulk dynamics are classical and with a variety of matter fields. Furthermore, by evolving backwards and forwards in time using timefolds with different boundary conditions, one can continue to iteratively expand the causal wedge, from apparent horizon to apparent horizon, all the way to the outermost extremal surface. Combining this result with ordinary HKLL leads to simple reconstructions of arbitrary operators in the outermost extremal wedge.

It is easy to see how this works in the case of Jackiw-Teitelboim gravity~\cite{Jackiw,Tei83} minimally coupled to a (classical) massless free scalar. In this setup, the matter factorizes into left and right movers; by changing the boundary conditions, we may ``absorb'' the right movers and turn off focusing on the future event horizon; this pushes the future event horizon backwards. We can then repeat the same procedure for the past horizon by evolving backwards in time; this will now push the past causal horizon backwards. The shift will have likely revealed additional left-movers, so the procedure needs to be iteratively repeated until it converges on a stationary bifurcate horizon. This is illustrated in Fig.~\ref{fig:turnoffsources}. The generalization to higher dimensions is significantly more technically challenging -- rather than removing sources of focusing entirely, it is more practical to ``stretch out'' the focusing over the causal horizon and so dilute its effect -- but the essential intuition is the same.

\begin{figure}
    \centering
    \includegraphics[width=0.9\textwidth]{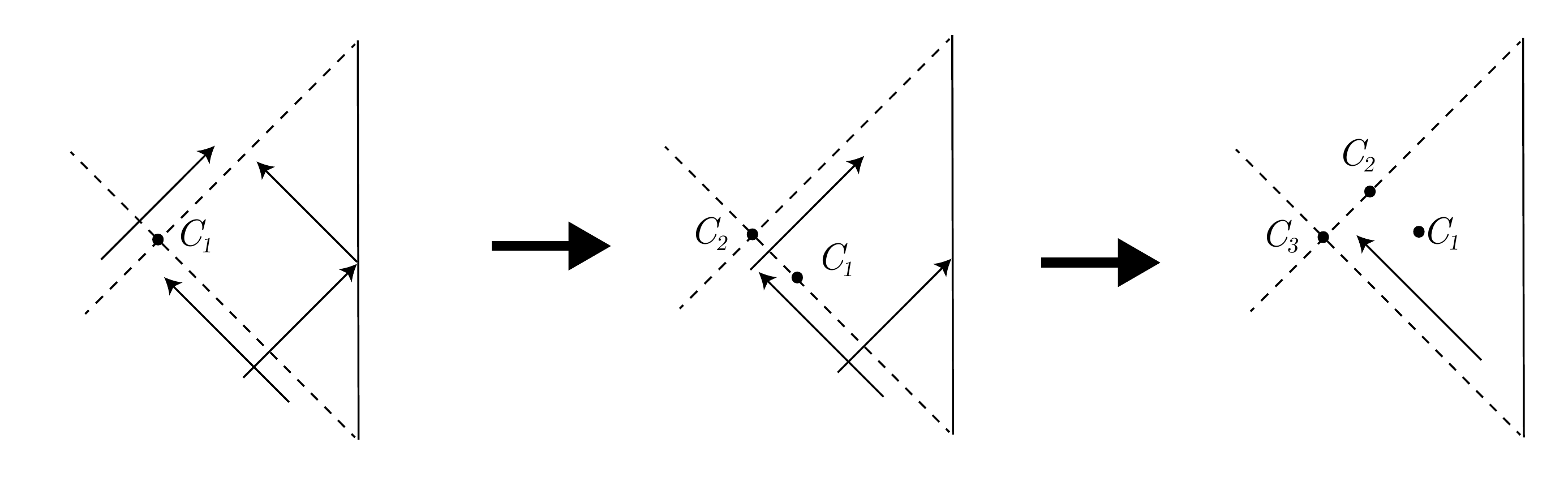}
    \caption{A caricature of the procedure used to push the causal wedge towards the appetizer in JT gravity coupled to a classical massless scalar. The leftmost panel is the original spacetime including left- and right-movers with reflecting boundary conditions. The causal surface of the right boundary is $C_{1}$. In the middle panel, the left movers have been turned off, which causes the future event horizon to shift inwards. The resulting causal surface is $C_{2}$, which is null-separated from $C_{1}$. The final panel shows that the right movers have been turned off, which causes the past event horizon to move inwards, shifting the causal surface to $C_{3}$. This shift reveals new left-movers in the causal wedge, which will have to be removed in subsequent zigzags along the past and future event horizons.}
    \label{fig:turnoffsources}
\end{figure}

The original motivation in \cite{EngWal17b} for attempting to expand the causal wedge using simple sources was to understand the holographic dual of the \emph{simple entropy}, defined as a maximization of the von Neumann entropy over all CFT density matrices with the same one-point functions -- with arbitrary time-ordered simple sources turned on after some initial time $t$ -- as the original CFT state. In other words, the simple entropy coarse-grains over all of the details of the state, except for simple observables that can be measured in the future of the initial time $t$. It was conjectured in \cite{EngWal17b, EngWal18} that the simple entropy is the boundary dual of the \emph{outer entropy}, a bulk quantity that coarse-grains over the geometry behind the outermost apparent horizon null-separated from the boundary at time $t$; it is equal to (one quarter of) the area of the apparent horizon. As a corollary of the results discussed above, we prove that this conjecture is indeed true whenever the bulk physics can be treated classically.

What if we generalize the definition of the simple entropy to allow  not just time-ordered insertions of simple operators, but insertions on arbitrary timefolds (and at arbitrary time)? In this case, there is no obstacle to seeing behind apparent horizons. By evolving the state backwards (and forwards) in time and then turning simple sources, it is possible to causally alter the spacetime near the apparent horizon, changing its location and ``seeing'' degrees of freedom that were originally hidden behind it. As per the discussion above, the first obstruction that cannot be bypassed in this way is the outermost extremal surface. Indeed, our results demonstrate that the simple entropy with arbitrary timefolds allowed is holographically dual to the area of the outermost extremal surface. Similarly, the density matrix $\rho_{\mathrm{coarse}}$ whose von Neumann entropy is the simple entropy with timefolds allowed reconstructs exactly the entire outermost extremal wedge and no more.  In fact, we can actually construct a complete spacetime in which the outermost extremal wedge is the entire entanglement wedge  of one connected asymptotic boundary obtained using the spacetime doubling procedure of~\cite{EngWal17b, EngWal18}; thus $\rho_{\mathrm{coarse}}$ is the actual CFT state  dual to the canonical purification as proposed in~\cite{EngWal17b, EngWal18} and proven in~\cite{DutFau19}.  See Fig.~\ref{fig:coarsegraining}.

\begin{figure}
    \centering
    \includegraphics[width=0.7\textwidth]{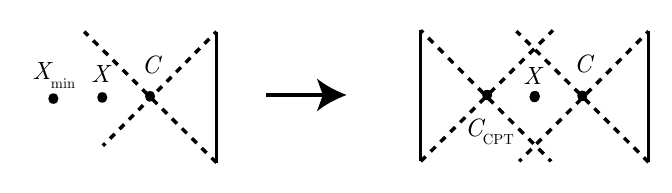}
    \caption{The coarse-graining procedure of~\cite{EngWal17b, EngWal18} as applied to the outermost extremal surface. The spacetime behind the outermost extremal surface $X$ is discarded and replaced with a CPT conjugate of the outermost extremal wedge. The rest of the spacetime is generated by standard Cauchy evolution.}
    \label{fig:coarsegraining}
\end{figure}

An immediate application of our result is then the construction of the CFT dual to the simple wedge in the final spacetime where the causal and outermost extremal wedges coincide. That is, this dual can be produced from $\rho_{\mathrm{coarse}}$ via a set of simple operations, with the dual bulk result being a two-sided black hole in which the bifurcation surface is extremal. The significance of this statement is manifold: we prove that the causal and entanglement wedges coincide if and only if the CFT state has a local modular Hamiltonian, which shows that finite time measurements cannot tell that the modular flow generated by the simple state is not local. In spacetimes with sufficient isotropy, the simple wedge CFT modular flow is in fact exactly local. This is analogous to a type of no-hair theorem: the set of holographic black holes with a stationary bifurcation surface is identical to the highly limited set of states with local modular Hamiltonians in the dual CFT.


From the perspective of holographic complexity, we may therefore interpret the absence of a Python’s lunch in the dual theory as the CFT state being related by a simple circuit to a rather special state with local modular flow (or at least indistinguishable from local in finite time). The world, it would seem, is rarely simple; pythons are ubiquitous. An explicit example of how a python might spring on an unsuspecting holographer in what would prima facie appear to be a python-less spacetime will be provided in our upcoming work~\cite{unexpected}. 

The paper is structured as follows. In Section~\ref{sec:SimpleWedge}, we define the outermost extremal wedge and the simple wedge, and we prove that the former is well-defined. In Section~\ref{sec:JT} we showcase our procedure for the simple case of JT gravity coupled to a massless (classical) scalar. In Section~\ref{sec:perturbation}, we prepare the perturbation that moves the causal horizon backwards along a the future event horizon in higher dimensional gravity (with arbitrary, null energy condition-satisfying matter), and we prove that the required perturbation satisfies the constraint equations. Section~\ref{sec:ZigZag} describes the zigzag portion of the procedure and completes the proof. We discuss the implications of our results, from the dual to the area of the outermost extremal surface to the nature of the simple state, in Section~\ref{sec:appetizer}. We finish with a discussion of generalizations and other implications in Section~\ref{sec:disc}.

\paragraph{Assumptions and Conventions:} The bulk spacetime $(M,g)$ is assumed to be classical with the dynamics governed by the Einstein field equation, i.e. we work in the large-$N$ and large $\lambda$ limit of AdS/CFT except where otherwise stated. We will assume the AdS analogue of global hyperbolicity~\cite{Wal12}. We also assume that the initial spacetime under consideration is one that satisfies the Null Energy Condition (NEC):
\begin{align}
T_{a b}k^{a} k^{b} \geq 0 
\end{align}
where $T_{ab}$ is the stress energy tensor and $k^a$ is any null vector. We will demonstrate that our perturbations of the spacetime maintain the NEC. All other conventions are as in~\cite{Wald} unless otherwise stated. 

\begin{itemize}

\item We shall use $J^{\pm}$ to refer to the bulk causal future and past and $I^{\pm}$ to refer to the bulk chronological future and past. Given a closed achronal set $S$, we use $D[S]$ to denote its domain of dependence, which we shall take to contain its boundary, as in~\cite{Wald}. $D^+[S]$ and $D^-[S]$ refer to the future and past components of the domain of dependence.

\item Hypersurfaces will refer to codimension-one embedded submanifolds of arbitrary signature.

\item By a ``surface'' we will always mean an achronal, codimension-two embedded submanifold which is Cauchy-splitting~\cite{BouEng15b}. Two surfaces $\sigma_1$ and $\sigma_2$ are homologous whenever there exists a hypersurface $H$ such that $\partial H =\sigma_1 \cup \sigma_2$. We will be primarily interested in surfaces homologous to (partial) Cauchy slices of the asymptotic boundary (CFT (sub)regions). 

\item Let $\Sigma$ be a Cauchy slice containing a surface $\sigma$ homologous to a boundary (sub)region $R$. By definition, $\sigma$ splits $\Sigma$ into two disjoint components that we will denote $\mathrm{Int}_\Sigma[\sigma]$ and $\mathrm{Out}_\Sigma[\sigma]$, where the conformal completion of the latter contains the boundary subregion $R$. We define $W_\sigma \equiv D[\mathrm{Out}_\Sigma[\sigma]]$, the outer wedge of $\sigma$. Similarly, we define $I_\sigma \equiv D[\mathrm{Int}_\Sigma[\sigma]]$, the inner wedge of $\sigma$. See also \cite{EngWal18}.

\item For a smooth surface $\sigma$ homologous to a boundary (sub)region, we denote by $k^a$ and $\ell^a$ the unique future-directed orthogonal null vector fields on the $C^{1}$ subsets of $\sigma$ pointing towards $\mathrm{Out}_\Sigma[\sigma]$ and towards $\mathrm{Int}_\Sigma[\sigma]$ respectively.

\item We define $\partial^+ W_{\sigma} = \partial D^+[\mathrm{Out}_\Sigma[\sigma]]$ and $\partial^- W_{\sigma} = \partial D^-[\mathrm{Out}_\Sigma[\sigma]]$. When $\sigma$ is smooth, $\partial^+ W_{\sigma}$ and $\partial^- W_{\sigma}$ can be constructed by firing null congruences starting from $k^a$ and $\ell^a$, terminating the congruence at caustics and non-local self-intersections~\cite{Wald, AkeBou17}.

\item Given any orthogonal null vector field $k^a$ on surface $\sigma$, $\theta_{(k)}$ denotes the expansion of $\sigma$ along $k^a$.\footnote{$k^{a}$ is not uniquely defined in places where $\sigma$ is not $C^{1}$. However, the expansion still has a definite sign, which can be computed via a limiting procedure.} We will refer to the following types of $\sigma$ based on its expansions:
\begin{itemize}


\item A compact $\sigma$ is trapped if $\theta_{(k)} < 0$ and $\theta_{(\ell)} < 0$, and marginally trapped if $\theta_{(\ell)} < 0$ and $\theta_{(k)}=0$.

\item $\sigma$ is \emph{extremal} if $\theta_{(k)}=0$ and $\theta_{(\ell)}=0$. By linearity, $\sigma$ is then stationary under deformations along any direction.

\end{itemize}

\item The future and past causal horizons associated to any boundary spacetime region $R \subset \mathscr{I}$ are defined as $\partial J^-[R]$ and $\partial J^+[R]$ respectively. By convention, we use $k^a$ and $\ell^a$ to refer to the generators of the future and past horizons respectively. More specifically, the future and past \textit{event} horizons are defined as $\mathscr{H}^+ \equiv \partial J^-[\mathscr{I}]$ and $\mathscr{H}^- \equiv \partial J^+[\mathscr{I}]$. The causal surface can be defined as $C \equiv \mathscr{H}^+ \cap \mathscr{H}^-$ and the causal wedge is $\mathscr{W}_C \equiv J^{+}[\mathscr{I}]\cap J^{-}[\mathscr{I}]$\cite{HubRan12}.\footnote{Note that although the term causal wedge is common in the literature, unlike the entanglement wedge it does not refer to a domain of dependence.} An important result in general relativity -- which follows from NEC and cosmic censorship -- is that future causal horizons satisfy an area law: the areas of their cross sections do not decrease as we move the cross section to the future~\cite{Haw71,HawEll}. In particular, this means that any congruence of null generators on a future horizon has nonnegative expansion. By time-reversal symmetry, a ``reverse'' area law holds for past horizons.

\item We define the terminated horizons $\mathscr{H}^+_C \equiv \mathscr{H}^+ \cap J^{-}[\mathscr{I}]$ and $\mathscr{H}^-_C \equiv \mathscr{H}^- \cap J^{-}[\mathscr{I}]$. These are natural definition for us since we are interested in perturbations of $\mathscr{H}^+$ and $\mathscr{H}^-$ caused by causal boundary sources.

\end{itemize}

\section{Which Wedge?} \label{sec:SimpleWedge}

Three bulk regions are under consideration here: the \textit{outermost extremal wedge}, the \textit{causal} wedge, and the \textit{simple} wedge. We will ultimately argue that the outermost extremal wedge is in fact the simple wedge, but in order to avoid subscribing to our own conclusions before we have demonstrated them, we introduce terminology that distinguishes between the two.

We will argue for the equivalence between the outermost extremal and simple wedges by showing that simple operations and sources (together with a finite number of time-folds) are sufficient to shift the causal wedge so that it comes arbitrarily close to coinciding with the outermost extremal wedge. While our primary results are for compact extremal surfaces, many of our intermediate results remain valid for boundary-anchored surfaces. In Sec.~\ref{sec:bdyanchored} we will discuss in more detail the extent to which our results apply to the latter case. 

We have already defined the more familiar causal wedge in the introduction. Let us now give a precise definition of the outermost extremal wedge and the simple wedge.

Intuitively, the extremal wedge is defined as the analogue of the entanglement wedge for the outermost extremal surface -- be it minimal or not. So before defining the outermost extremal wedge, we must prove that a unique outermost extremal surface exists in the first place:

\begin{prop}\label{prop-X}
If there exists more than one extremal surface homologous to a connected component of $\mathscr{I}$, then exactly one is outermost; i.e. there exists an extremal surface $X$ contained in the outer wedge $W_{X'}$ of all other extremal surfaces $X'$ homologous to $\mathscr{I}$. 
\end{prop}

We will prove this proposition using a series of three lemmas:

\begin{lem}
Given two surfaces $\sigma_1$ and $\sigma_2$ homologous to $\mathscr{I}$, $W_{\sigma_1} \cap W_{\sigma_2}$ is a domain of dependence.
\end{lem}

\begin{proof}

Let $\Sigma$ be a Cauchy slice containing $\sigma_1$. We define a new surface $\Sigma' = \partial ( (J^+[\Sigma] \cup J^+[\sigma_2]) \cap \overline{I^{-}[\sigma_2]})$, where $\overline{I^-[\sigma_2]} \equiv M - I^-[\sigma_2]$. Since every inextensible timelike curve at some point in its past is outside $J^+[\Sigma] \cup J^+[\sigma_2]$, but eventually ends up in $J^+[\Sigma] \cap \overline{I^{-}[\sigma_2]}$, and since no timelike curve can exit $(J^+[\Sigma] \cup J^+[\sigma_2]) \cap \overline{I^{-}[\sigma_2]})$ after entering it, $\Sigma'$ is Cauchy. Note that, despite appearances, the definition of $\Sigma'$ is invariant under time-reversal symmetry. 
 
We will now show that $W_{\sigma_1} \cap W_{\sigma_2}$ is the domain of dependence of $H = \Sigma' \cap W_{\sigma_1} \cap W_{\sigma_2}$. Any causal curve intersecting $W_{\sigma_1} \cap W_{\sigma_2}$ needs to intersect $\Sigma \cap W_{\sigma_1}$ either (i) in $W_{\sigma_2} - \partial W_{\sigma_2}$ or (ii) outside of it . In case (ii), the causal curve needs to leave $J^{-}[\sigma_2]$ in $W_{\sigma_1}$ after intersecting $\Sigma$ or enter $J^+[\sigma_1]$ in $W_{\sigma_1}$ before intersecting $\Sigma$. Therefore, in both cases i and ii we conclude that the causal curve intersects $H$.

\end{proof}

Let $V$ now be a domain of dependence that intersects $\mathscr{I}$. Then there must exist an ``edge'' surface $\sigma$ homologous to $\mathscr{I}$ such that $V = W_\sigma$. More precisely, $\sigma$ can be defined as the set of points $p \in \partial V$ such that in any small neighborhood of $p$ any inextensible timelike curve crossing $p$ only intersects $V$ at $p$.

\begin{lem}
Let $V = W_{\sigma_1} \cap W_{\sigma_2}$ and let $\sigma$ be the edge of $V$ as defined above. Then, $\sigma \subset \sigma_1 \cup \sigma_2 \cup (\partial^+W_{\sigma_1} \cap \partial^- W_{\sigma_2}) \cup (\partial^-W_{\sigma_1} \cap \partial^+W_{\sigma_2})$.
\end{lem}

\begin{proof}
Clearly, the edge is contained in $\partial V \subset \partial W_{\sigma_1} \cup \partial W_{\sigma_2}$. Say some point $p$ in the edge of $V$ were in $\partial^+W_{\sigma_1}-\sigma_1$. Then, every timelike curve crossing $p$ leaves $W_{\sigma_1}$ (and therefore $V$) in $I^+(p)$ and enters $W_{\sigma_1}$ in $I^-(p)$. Furthermore, for $p$ to be in the edge of $V$ it needs to be in $W_{\sigma_2}$, but $I^{-}(p)$ must not intersect $W_{\sigma_2}$. Therefore, $p$ needs to be in $\partial^-W_{\sigma_2}$. Together with the time reverse of this argument and also switching $\sigma_1$ and $\sigma_2$, we conclude that the edge of $V$ must be contained in $\sigma_1 \cup \sigma_2 \cup (\partial^+W_{\sigma_1} \cap \partial^-W_{\sigma_2}) \cup (\partial^-W_{\sigma_1} \cap \partial^+W_{\sigma_2})$.
\end{proof}

Lastly, we will state a lemma from Sec. 2.2 of \cite{BouSha21}, without providing the proof (see also Appendix B of \cite{BroGha19} for a similar discussion). The lemma assumes the existence of a stable maximin surface~\cite{Wal12} in any domain of dependence. 

\begin{lem}\label{lem-restrictedmaximin} If $\sigma$ is a surface homologous to $\mathscr{I}$ satisfying $\theta_{(k)} \leq0$ and $\theta_{(\ell)} \geq 0$, then there exists an extremal surface $Y \subset W_\sigma$ homologous to $\mathscr{I}$.
\end{lem}

Let us provide some intuition for Lemma \ref{lem-restrictedmaximin}. The restricted-maximin prescription returns a surface in $W_\sigma$ which is homologous to $\mathscr{I}$ and is minimal on some Cauchy slice of $W_\sigma$. If $\sigma$ satisfies $\theta_{(k)} <0$ and $\theta_{(\ell)} > 0$, then the maximin surface cannot intersect $\sigma$, since its area would get smaller under deformations away from such intersections. In \cite{BouSha21}, it was further shown that the max property of the maximin surface prohibits intersections with $\partial W_\sigma - \sigma$. This shows that the maximin surface is in the interior of $W_\sigma$ and thus extremal. Furthermore, it was argued that even when the inequalities are not strict, there still exists an extremal surface homologous to $\mathscr{I}$ in $W_\sigma$ even though in the surface might lie on $\partial W_\sigma$.

We are now ready to prove Proposition \ref{prop-X}:

\begin{proof}
In this proof any ``surface'' will mean a surface that is homologous to $\mathscr{I}$. Let an extremal surface $X$ be called \textit{exposed} if there does not exist any extremal surface $Y \neq X$ such that $Y \subseteq W_X$. Let us first argue that there must always exist at least one exposed surface. Define a partial ordering on extremal surfaces by declaring $X \geq Y$ if and only if $W_X \subseteq W_Y$. Note that exposed surfaces would correspond to maximal elements with respect to this partial order, while an outermost extremal surface would be a greatest element. Upper bounds exist for any chain because monotonicity and boundedness (there are no extremal surfaces near asymptotic infinity) ensure that any sequence $\{X_n\}$ of extremal surfaces with $W_{X_{n+1}} \subseteq W_{X_{n}}$ converges to an extremal surface $X_\infty$ with $W_{X_\infty} \subseteq W_{X_{n}}$ for any finite $n$. Hence, by Zorn's lemma, at least one maximal element, i.e. exposed surface, exists. 

Now suppose, by way of contradiction, that there exists an exposed surface $X_1$ that is not outermost, i.e. there exists some other extremal surface $X_2$ such that $W_{X_1} \nsubseteq W_{X_2}$. Let $V = W_{X_1} \cap W_{X_2}$ (by definition then $V \subset W_{X_1}$). By Lemma 2 and 3,  $V = W_{\sigma}$ for some surface $\sigma \subseteq X_1 \cup X_2 \cup (\partial^+W_{X_1}\cap \partial^-W_{X_2}) \cup (\partial^-W_{X_1}\cap \partial^+W_{X_2})$.
Let us consider each component of the set to which $\sigma$ belongs separately. The first two are subsets of extremal surfaces and hence have zero expansion. By focusing, the null hypersurface $\partial^+W_{X_1}$ and $\partial^+W_{X_2}$ are non-expanding towards the future, while the null hypersurface $\partial^-W_{X_1}$ and $\partial^-W_{X_2}$ are non-expanding towards the past. Therefore, $\sigma$ satisfies $\theta_{(k)} \leq 0$ and $\theta_{(\ell)} \geq 0$\footnote{This $\sigma$ may have kinks in which case the expansions are technically not well-defined. However, we expect that all of our results in this section generalizes easily to such surfaces using standard geometric techniques. Intuitively, by slightly smoothing out the kinks we can get a new infinitesimally nearby surface with very large expansions of the desired signs around the kinks.}, and hence by Lemma \ref{lem-restrictedmaximin}, there exists an extremal surface $Y$ contained in $W_\sigma = W_{X_1} \cap W_{X_2}$. The existence of this surface means that $X_1$ is not exposed, giving our desired contradiction.

\end{proof}

As an aside, note that extending the last part of this argument to quantum extremal surfaces, under assumption of the quantum focusing conjecture requires a small amount of extra work because the definition of quantum expansion is nonlocal. However strong subadditivity is enough to ensure that the quantum expansion of $\sigma$ satisfies the desired inequalities.

It is now easy to define the outermost extremal wedge:

\begin{defn} Let $X$ be the outermost extremal surface for a connected component of $\mathscr{I}$. The outermost extremal wedge $W_{X}$ is the outer wedge of $X$.
\end{defn}

Next we would like to define the simple wedge. Conceptually, the simple wedge is the largest bulk region that can be reconstructed from the near boundary state of the bulk fields using exclusively the bulk equations of motion. Consider some boundary state $\rho$ whose bulk dual $(M,g)$ we would like to reconstruct and evolve it to the far past and future with some Hamiltonian. In the classical regime, following Eq. \eqref{eq:HKLL}, HKLL prescribes the values of the bulk fields in $\mathscr{W}_C$ from the set of one-point functions of their corresponding local boundary operators on $\mathscr{I}$. In fact, the bulk equations of motions are sufficient to reconstruct the maximal Cauchy development of $\mathscr{W}_C$ -- which we can denote by $W_C$.

Cosmic censorship in general~\cite{Pen69, GerHor79} and causal wedge inclusion in particular~\cite{Wal12, HeaHub14, EngFol20} guarantees that the causal wedge contains no extremal surfaces in its interior: $W_{C}\subseteq W_{X}$. It is therefore impossible to reconstruct the region behind $X$ causally. HKLL alone, however, appears to prima facie fail at an earlier stage: the non-standard Cauchy evolution appears to stop short of recovering the gap region between $W_{C}$ and $W_{X}$, which is generically non-empty.

What if we evolve $\rho$ using a different Hamiltonian? Consider turning on a set of CFT operators at various times during the evolution. This would ``extract'' a new set of one-point functions from $\rho$ and therefore has the potential to expand the reconstructible region. In keeping with our philosophy of the simple wedge, we must restrict to sources that have a (semi-)classical bulk dual. Therefore, following \cite{EngWal17b, EngWal18}, we refer to boundary sources as simple if the bulk fields that they produce propagate causally into the bulk from the boundary insertion point -- and restrict to such sources henceforth. The change in time evolution when such simple sources are applied within some boundary time interval $[t_i, t_f]$ is given by the following time-ordered operator:

\begin{align}\label{eq-source}
E = \mathcal{T} exp \left[ -i \int_{t_i}^{t_f} dt'~J(t') \mathcal{O}_{J}(t')  \right],
\end{align}
where $J(t)$ is a simple source and $\mathcal{O}_J$ is its corresponding simple operator. Note that $\mathcal{O}_J$ might involve spatial integrals of local boundary operators $\mathcal{O}(t',x')$. An example of a simple operator is a spatial integral of a single-trace operator of the boundary gauge theory.

Adding $E$ to the evolution, say in a future-directed timefold, changes the spacetime from $M$ to some $M'$. By causality, $M - J^{+}[t_i]\footnote{Here $t_i$ stands for a timeslice of $\mathscr{I}$ geometrically.} \subset M'$. In particular, the perturbation to the spacetime is localized away from the past event horizon $\mathscr{H}^-$. However, sources like Eq. \eqref{eq-source} will typically change where the new future event horizon $\mathscr{H}^+$ intersects $\mathscr{H}^-$. In particular, suppose that we find simple sources that ``expand'' the causal wedge, i.e. place the new causal surface $C'$ in the future of $C$ on $\mathscr{H}^-$. Said in the CFT language, the new set of one-point data reconstructs a $\mathscr{W}'_{C}$ that contains $\mathscr{H}^-_C$. Furthermore, knowing the bulk equations of motion and the original Hamiltonian, we can reconstruct the Cauchy development of $\mathscr{H}^-_{C'}$, a wedge in the \emph{original} spacetime that contains $W_C$ as a proper subset.

It is natural to define the simple wedge according to the maximal success of this procedure:

\begin{defn}
The \textit{simple} wedge is the maximal bulk region that can be reconstructed from \textit{simple operators} acting on the dual CFT state, with the inclusion of simple sources and timefolds.
\end{defn}

Although we have defined the simple wedge in the context of classical field theory in the bulk, it is important to note that HKLL can reconstruct the quantum state of the bulk propagating in the causal wedge at each step. For bulk fields in the $1/N$ expansion, Eq. \eqref{eq:HKLL} provides the dictionary between local bulk operators and simple CFT operators, realizing this reconstruction in the Heisenberg picture.

Finally, we close this section by relating the causal and outermost extremal wedges. Intuitively, the causal surface should coincide with the extremal surface if and only if there is no focusing whatsoever on the horizons. However, because the extremal wedge is defined as a domain of dependence and the causal wedge is defined in terms of causal horizons, it does not immediately follow that the two must coincide in the absence of focusing. To reassure the reader, we prove the following lemma:

\begin{lem}
Let $X$ be the outermost extremal surface homologous to one or more connected components of $\mathscr{I}$. $W_{X}=\mathscr{W}_{C}$ if and only if $X$ is a bifurcation of stationary horizons (and thus $X=C$).
\end{lem}

\begin{proof}
If $\mathscr{W}_{C}=W_{X}$, then $\partial \mathscr{W}_{C}=\partial W_{X}$. If both wedges had been defined in terms of domains of dependence, it would immediately follow that $C=X$. However, since $\mathscr{W}_{C}$ is not defined as a domain of dependence, we have to work a little harder. The component of $\partial \mathscr{W}_{C}$ which is spacelike separated to every point in Int$[\mathscr{W}_{C}]$ is identical to the component of $W_{X}$ which is spacelike-separated to every point in Int$[W_{X}]$. By definition of the causal wedge, the former is $ C$. Since $W_{X}$ is generated by the domain of dependence of a hypersurface $H$ whose boundary in $M$ is $X$, $X$ is exactly the set of points that are spacelike separated from every point in Int$[W_{X}]$. This immediately shows that $C=X$. Because $\theta_{(n)}[X]=0$ for all $n^{a}$ in the normal bundle of $X$, we find that $\theta_{(\ell)}[C]=0=\theta_{(k)}[C]$. By the NCC, a future (past) causal horizon can only have vanishing expansion on a slice if it has vanishing expansion everywhere to the future (past) of that slice. So $\partial \mathscr{W}_{C}$ and subsequently $\partial W_{X}$  is generated by stationary horizons, and $X$ is a bifurcate stationary horizon. 

If $X$ is a bifurcate stationary horizon, then results of~\cite{BouEng15b, AkeBou17} immediately imply that $\partial W_{X}$ is the union of two truncated stationary horizons $\mathscr{H}^{\pm}$. Since $\mathscr{H}^{+}\cap \partial M = \partial D^{+}$, $\mathscr{H}^{+}$ is a past-directed null congruence fired from $i^{+}$. By the theorems of~\cite{BouEng15b}, $\partial J^{-}[\mathscr{I}]=\mathscr{H}^{+}$ up to geodesic intersections. Since $\mathscr{H}^{+}$ is stationary, it has no intersections, so $\partial J^{-}[\mathscr{I}]=\mathscr{H}^{+}$. Similarly $\partial J^{+}[\mathscr{I}]=\mathscr{H}^{-}$. Thus $\partial W_{X}=\partial\mathscr{W}_{C}$ and $W_{X}=\mathscr{W}_{C}$ by the homology constraint. 
\end{proof}
\noindent Note that this result remains valid for quantum extremal surfaces assuming the quantum focusing conjecture and a suitable generalization of AdS hyperbolicity.

\section{Two Dimensions}\label{sec:JT}
Let us illustrate our iterative procedure for removing matter falling across the past and future causal horizons in a simple toy model of JT gravity minimally coupled to a massless scalar field $\varphi$ (with no direct coupling between the dilaton $\Phi$ and $\varphi$). The absence of propagating degrees of freedom of the gravitational field as well as the factorization of the bulk matter into left- and right-movers are simplifications that naturally do not generalize to higher dimensions; nevertheless the procedure itself is well-illustrated in this setting, which we include for pedagogical reasons. The additional complications introduced in higher dimensions are resolved in subsequent sections.

Due to focusing resulting from the scalar field $\varphi$, the bifurcation surface will generically not be extremal, i.e.
\be
\partial_{n}\Phi|_{C_{1}}\neq 0
\ee
where $n^{a}$ is some vector normal to the causal surface $C_{1}$ (in particular, for null $n^{a}$ future-outwards directed, this would be positive; similarly for a time-reverse). As a consequence of the highly limited number of degrees of freedom in the problem, the extremality failure can only be a result of focusing: the future causal horizon will experience focusing due to the $\varphi$ left-moving modes and the past causal horizon will experience focusing due to the $\varphi$ right-moving modes. Our procedure instructs us to first remove the source of focusing of the future horizon by modifying the boundary conditions of $\varphi$, which we can easily do by implementing absorbing boundary conditions for the right movers in order to remove all the left-moving modes. This removes focusing from the future causal horizon, which pushes the horizon deeper into the bulk. As a consequence, the new causal surface (which is now marginally trapped) -- let us call it $C_{2}$ -- is pushed further along the past event horizon, which remains unmodified by this procedure. This first step is illustrated in the second panel of Fig.~\ref{fig:turnoffsources}. 

To turn off focusing of the past horizon, we evolve backwards in time, imposing boundary conditions that remove the right-movers. The past event horizon moves backwards, and the new causal surface $C_{3}$ (which is now marginally anti-trapped) is displaced from $C_{2}$ along the future event horizon of $C_{2}$. However $C_{3}$ is not necessarily extremal since shifting the past causal horizon reveals a part of the spacetime that was previously not included in the causal wedge: in particular, new left-moving modes can now appear in the causal wedge. This is illustrated in the third panel of Fig.~\ref{fig:turnoffsources}. This piecewise-null zigzag procedure thus shifts the causal surface deeper into the bulk; we may simply repeat the zigzag iteratively. 

In classical gravity, the focusing theorem and cosmic censorship (or strong asymptotic predictability) together guarantee that no extremal surface is ever in causal contact with $\mathscr{I}$: so the zigzag procedure can never modify an extremal surface nor move the causal surface deeper than any extremal surface. Thus the outermost extremal surface is an upper bound on the success of the procedure. Our goal, of course, is to show that this upper bound is in fact attained. 
\begin{figure}
    \centering
    \includegraphics[width=0.4\textwidth]{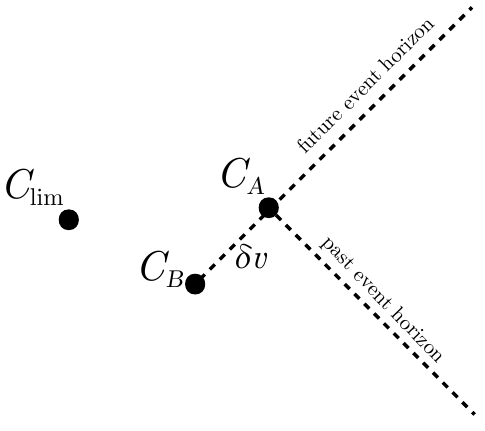}
    \caption{An illustration of the approach to the limit point $C_{\mathrm{lim}}$, where $C_{A}$ and $C_{B}$ are infinitesimally close to the limiting point.}
    \label{fig:limiting}
\end{figure}

Because the success of this procedure is bounded by the outermost extremal surface and because furthermore the procedure moves the surface monotonically inwards, a limiting causal surface $C_{\mathrm{lim}}$ exists, and the corresponding causal wedge does not intersect any extremal surface. We will now argue that $C_{\mathrm{lim}}$ is in fact extremal. Let $C_{A}$ be a causal surface obtained via iterative zigzags which is infinitesimally close to $C_{\mathrm{lim}}$; without loss of generality we may take $C_{A}$ to be in the marginally trapped portion of the zigzag (i.e. the left-movers had just been removed). Let $v$, $u$ be the affine parameters along the future and past event horizons, respectively, in the spacetime in which $C_{A}$ is the causal surface. By construction, $C_{A}$ has no expansion along the future event horizon:
\begin{equation}
 \frac{\partial}{\partial v}\Phi|_{C_{A}} =0,  
\end{equation}
and because by assumption it is not identical to $C_{\mathrm{lim}}$, 
\begin{equation}
 \frac{\partial}{\partial u}\Phi|_{C_{A}} <0.
\end{equation}
Let $C_{B}$ be the causal surface obtained by removing the right-movers from the spacetime where $C_{A}$ is the causal surface. By construction
\begin{equation}
 \frac{\partial}{\partial u}\Phi|_{C_{B}} =0.
\end{equation}
By the zigzag procedure, $C_{A}$ and $C_{B}$ are null-separated along the ``old'' future event horizon: i.e. along the null congruence that is the future event horizon in the spacetime in which $C_{A}$ is the causal surface. Let $\delta v$ be the amount of affine parameter separating $C_{A}$ and $C_{B}$. See Fig.~\ref{fig:limiting} for an illustration. Since the points $C_{A}$ and $C_{B}$ must be infinitesimally close to one another (since both infinitesimally near $C_{\mathrm{lim}}$), the spacetime metric in that neighborhood may be approximated as locally flat instead of AdS$_{2}$; using $u$ and $v$ as coordinates:
\be
ds^{2}=-2dudv.
\ee
In these coordinates, it is trivial to relate $\delta v$ to the change in $\partial_{u}\Phi$ along $v$. In particular, we may bound it from below:
\be
\delta v\geq \frac{\partial_{u} \Phi}{\partial_{v}\partial_{u}\Phi|_{\mathrm{max}}},
\ee
where $\partial_{v}\partial_{u}\Phi|_{\mathrm{max}}$ is the maximum value of $\partial_{v}\partial_{u}\Phi$ on the $\delta v$ interval. Similarly defining $\delta u$ as the null separation between $C_B$ and the next causal surface after again removing left movers, we obtain an analogous bound:
\be
\delta u\geq \frac{\partial_{v} \Phi}{\partial_{u}\partial_{v}\Phi|_{\mathrm{max}}}.
\ee
Under the assumption that $\partial_{u}\partial_{v}\Phi$ is bounded from above in this neighborhood (which we generically expect to be true), $\delta v$ and $\delta u$ approach zero no slower than $\partial_{u}\Phi$ and $\partial_{v}\Phi$ approach zero: thus $C_{\mathrm{lim}}$ must be extremal. Because focusing arguments ensure that the causal wedge is always contained in the outermost extremal wedge, it must be the outermost extremal surface.

\section{The Perturbation} \label{sec:perturbation}

Our task in higher dimensional gravity is now clear: we must find a perturbation that removes focusing from the causal horizons (without violating the null energy condition anywhere in the perturbed spacetime), thus shifting the causal wedge deeper in. What kind of perturbation $\delta g$ would move the causal surface towards rather than away from the appetizer? On a heuristic level, we are looking to open up the bulk lightcones so that more of the bulk is in causal contact with the asymptotic boundary. In searching for such a perturbation, we may build on the intuition of the boundary causality condition~\cite{EngFis16}, which states that the inequality
\be
\int\delta g_{ab}k^{a}k^{b}\geq 0
\ee
(where the integral is over a complete null geodesic with generator $k^{a}$) is equivalent to demanding that perturbations $\delta g$ of pure AdS source focusing (as opposed to defocusing). Here we are looking to do the opposite: we are looking to undo focusing, so it makes sense to look for a perturbation that satisfies an opposite inequality, with $\delta g_{kk}<0$ everywhere on $\mathscr{H}^+_C$. It is a priori not clear that it is possible to find a perturbation that simultaneously satisfies this inequality and also results in a spacetime that solves the Einstein equation. To prove this point, we must show that the requisite $\delta g$ solves the characteristic constraint equations on the event horizon.

In this section, we will prove that as long as the causal surface is not marginally outer trapped -- i.e. as long as $\theta_{(k)}[C]\neq 0$, it is possible to find exactly such a perturbation that (1) satisfies the characteristic constraint equations on the causal horizon and (2) shifts the causal surface deeper into the bulk. The procedure is roughly as follows: we prescribe a $\delta g$ deformation on $\mathscr{H}^+$; some elements of this $\delta g$ resemble the ``left stretch'' construction~\cite{BouCha19, BouCha20} involving a rescaling of the generators of certain achronal null hypersurfaces -- intuitively, this dilutes the infalling content on $\mathscr{H}^+$ and in turn reduces focusing. We then demonstrate that the gravitational constraints on $\delta g$ along with boundary conditions fix the requisite components of $\delta g$ in such a way that the perturbed spacetime has a larger causal wedge. In Sec. \ref{sec:ZigZag}, we will argue how repeating these perturbations pushes in the causal surface up to an apparent horizon in a given timefold, and to the appetizer using several timefolds.

We will call the generators of the future and past event horizons $k^{a}$ and $\ell^{a}$ respectively.~\footnote{These will agree with $k^{a}$ and $\ell^{a}$ on $C$ whenever $C$ is $C^{1}$.} We will extend $\ell^a$ to the entire spacetime by picking a smooth Cauchy foliation $\{{\cal C}_{\alpha}\}_{\alpha}$ of $\mathscr{I}$ and defining $\ell^{a}$ to be the bulk generators of $\partial I^{+}[{\cal C}_{\alpha}]$. This defines a null foliation of $J^{+}[\mathscr{I}]$ by past causal horizons; The past event horizon, which $C$ lies on, is a member of this foliation.

We adopt the coordinate and gauge choices of~\cite{EngWal18}: first, we introduce double null coordinates $(u,v)$ \footnote{Calling these double-null coordinates is a slight abuse of notation, as we will only require $g_{vv}=0$ at the horizon $u=0$.}, where
\begin{equation}  k^{a} = \left (\frac{\partial }{\partial v} \right)^{a}  \ \ ; \ \  \ell^{  a}= \left (\frac{\partial }{\partial u} \right)^{a}.
\end{equation}

\begin{figure}
    \centering
    \includegraphics[width=0.6\textwidth]{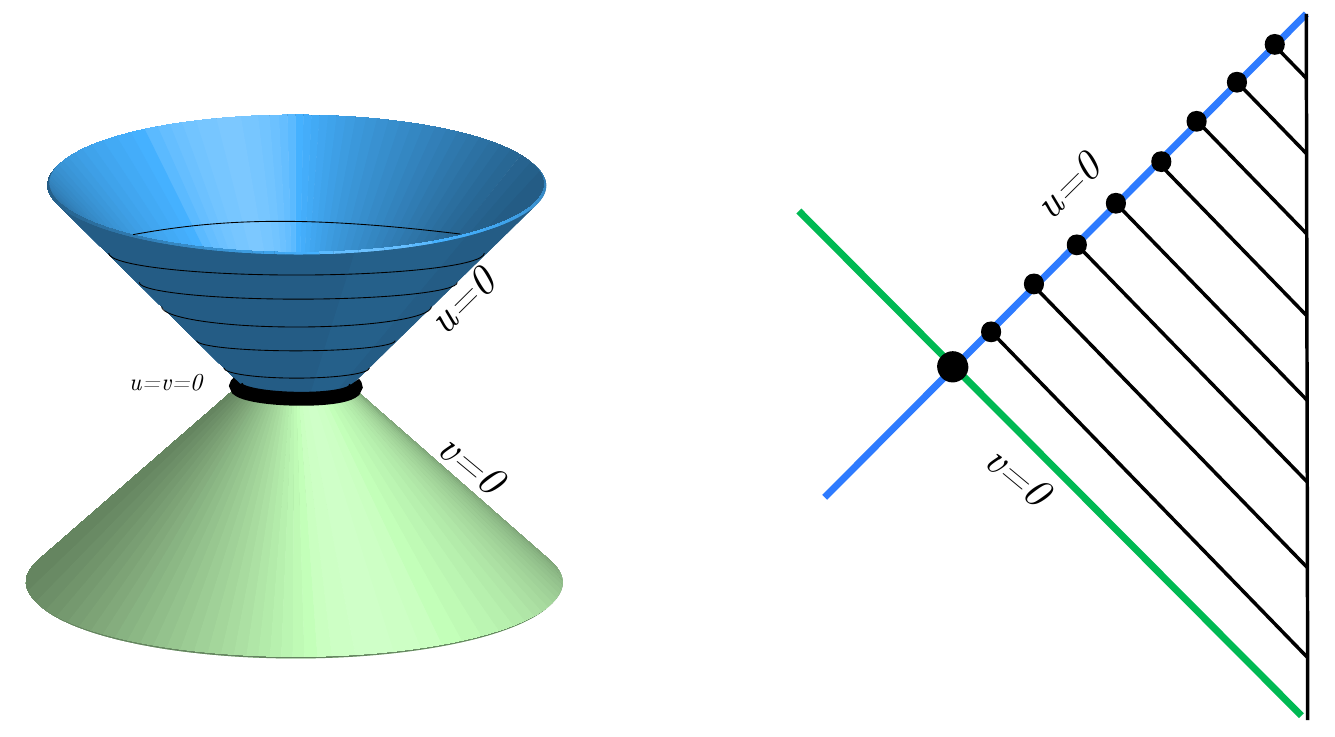}
    \caption{On the left: three-dimensional (left) illustration of the past ($v=0$) and future ($u=0$) event horizons, with a slicing of the future event horizon given by the intersection of $u=0$ with past causal horizons originating from complete slices of the $\mathscr{I}$. On the right: a conformal diagram illustrating the same.}
    \label{fig:AffineSlice}
\end{figure}

In these coordinates, the causal surface is at $u=v=0$, the future event horizon $\mathscr{H}^+$ is at $u=0$ and the past event horizon $\mathscr{H}^-$ is at $v=0$. See Fig.~\ref{fig:AffineSlice}. We can further fix the gauge in a neighborhood of $u=0$ so that the metric there takes the form:
\be
ds^2 = -2 du dv + g_{vv}dv^{2}+ 2 g_{vi} dv dy^i + g_{ij} dy^i dy^j
\ee
where $i,j$ denote the transverse direction on $\mathscr{H}^+$. At $u=0$ exactly, we of course require that $k^{a}=(\partial/\partial v)^{a}$ is null, so that:
\be
g_{vv}\rvert_{u=0}=0,
\ee
and we may further fix the gauge:
\be
g_{vi}\rvert_{u=0}=0,
\ee
but we cannot require these components to be \textit{identically} zero in a neighborhood of $u=0$, i.e. the derivatives may not vanish. For instance, the inaffinity of $k^a$ -- given by $\kappa_{(v)} = \frac{1}{2}\partial_u g_{vv}\rvert_{u=0}$ -- cannot be set to zero in general since we have independently fixed the $\ell^a$ vector field, and orthogonality to $\ell^a$ defines the constant (or affine) $v$ slices on $\mathscr{H}^+$.\footnote{The covariant definition of the inaffinity of $k^{a}$ is $k^a\nabla_a k^b=\kappa_{(k)} k^b$.}

The extrinsic curvature tensors of constant-$v$ slices are simple in this gauge
\bea
{B_{(v)}}_{ij} &= \frac{1}{2} \partial_v g_{ij}\rvert_{u=0},\\
{B_{(u)}}_{ij} &= \frac{1}{2} \partial_u g_{ij}\rvert_{u=0},\\
\chi_i &= \frac{1}{2} \partial_u g_{vi}\rvert_{u=0}.
\eea

\noindent where ${B_{(v)}}_{ij}$ and ${B_{(u)}}_{ij}$ denote the null extrinsic curvatures of constant $v$ slices and $\chi_i$ denotes the their twist. Since $\mathscr{H}^+$ is achronal, we can specify a new solution to the Einstein equation via a perturbative deformation of the metric on it, so long as the null constraint equations are satisfied. In particular, we consider the following perturbation on $\mathscr{H}^+_C (u=0, v\geq0)$:

\be \label{eq-perturb}
ds^2 = -2 du dv + (g_{vv}+\delta g_{vv}) dv^2 + 2(g_{vi}+\delta g_{vi}) dv dy^i + (g_{ij} + \delta g_{ij}) dy^i dy^j.
\ee
The perturbation components $\delta g_{vv}$, $\delta g_{vi}$, $\delta g_{ij}$ and their first $u$ derivatives are initial data that we can freely specify on the characteristic hypersurface $u=0$ so long as this data satisfies the null constraint equations\footnote{The null constraint equations are the equations of motion containing at most one $u$ derivative of the metric perturbations.}:
\bea
&\delta G_{vv} + \Lambda \delta g_{vv} = 8 \pi G \delta T_{vv} \label{eq-deltaRvv}\\
&\delta G_{uv} = 8 \pi G \delta T_{uv}\label{eq-deltaRuv}\\
&\delta G_{vi} +\Lambda \delta g_{vi} = 8 \pi G \delta T_{vi}\label{eq-deltaRvi}\\
&\hat{\delta G_{ij}} +\Lambda \hat{\delta g_{ij}} = 8 \pi G \hat{\delta T_{ij}}\label{eq-deltaRij}
\eea
where $\delta G_{a b}$ denotes the linearized perturbation of the Einstein tensor, $\Lambda$ is the cosmological constant, and the hat in the last equation denotes the traceless part. Note that the corresponding deformations of the stress energy tensor must be sourced by a perturbative modification to the matter fields that itself satisfies the fields' equations of motion on the background geometry. 

Thus we will need to prescribe $\delta g$ as well as a matter source for it. The latter is accomplished by a perturbative ``stretch''. A nonperturbative stretch is an exponential rescaling~\cite{BouCha19, BouCha20}:
\bea
&g'_{ij}(u=0,v,y)= g_{ij}(u=0, v e^{-s}, y)\label{eq-gijboost}\\
&T'_{vv}(u=0,v,y) = e^{-2s} T_{vv}(u=0, v e^{-s}, y)\label{eq-mattervvboost}\\
&\kappa'_{(v)}(u=0, v, y) = e^{-s}\kappa_{(v)}(u=0, v e^{-s}, y)\label{eq-kappaboost}
\eea
where prime denotes the transformed quantities. Our matter source will be obtained in the perturbative limit of this, by setting $e^{-s} = 1-\epsilon$, where $\epsilon\sim \mathcal{O}(\delta g)$ is the parameter controlling the expansion. Our choice for the full perturbation on $v \geq 0$, metric and matter, is then:
\bea
&\delta g_{vi}(u=0,v,y)=0 \label{eq-mainperturb1}\\
&\partial_u \delta g_{vv}(u,v,y)\rvert_{u=0} =2(1-\epsilon)\kappa_{(v)}(0, v(1-\epsilon), y) - 2\kappa_{(v)}(u=0, v, y)\label{eq-mainperturb2}\\
&\delta g_{ij}(u=0,v,y) = g_{ij}(u=0,v(1-\epsilon),y) - g_{ij}(u=0,v,y)\label{eq-gijepsilon}\\
&\delta T_{vv}(u=0,v,y) = (1-\epsilon)^2 T_{vv}(u=0,v(1-\epsilon),y) - T_{vv}(u=0,v(1-\epsilon),y)\label{eq-Tvvepsilon}
\eea
where in this linearized analysis we will only need to keep track of first order terms in $\epsilon$ and $\delta g$ (i.e. we will drop all terms of order $\epsilon \delta g$). Note that $\delta g_{vv}$, $\partial_u \delta g_{ij}$, and $\partial_u \delta g_{vi}$ are allowed to be non-zero. We will see that their values are constrained subject to the above restrictions.

Before we move forward with the analysis, we need to ask whether we can always obtain the stress tensor profile in Eq. \eqref{eq-Tvvepsilon}. This question is difficult to answer in broad generality. Therefore, from now on we restrict our matter sector to consist of a minimally coupled complex scalar field theory $\phi$ (with an arbitrary potential) coupled to some Maxwell field (or consider either separately), with Lagrangian density
\begin{align}
\mathcal{L}_{\text{matter}} = -\frac{1}{4} g^{ac} g^{bd} F_{ab} F_{cd} - g^{ab} \bar{\nabla}_a \phi \bar{\nabla}_b \phi^* - V(|\phi|^2)
\end{align}
where $F_{ab}$ is the field strength and $\bar{\nabla}$ denotes the covariant derivative with respect to the vector potential $A^a$. Then 
\be T_{vv} = 2\bar{\nabla}_v \phi \bar{\nabla}_v \phi^* + F_{iv} F^i_v,
\ee
Since both $\phi$ and $A_v$ are free initial data in the characteristic problem, we can simply generate the desired transformation by setting $\phi'(v,y) =\phi(v e^{-s},y)$ and $A'_v = e^{-s}A_v(v e^{-s},y)$, and $A_i' = A_i(v e^{-s},y)$.

We now proceed to prove that our choice of perturbation solves the null constraint equations with $\delta g_{vv}<0$ everywhere on $\mathscr{H}^+_C$. Because the unperturbed spacetime satisfies the NEC and no new matter terms are introduced by the perturbation, the perturbed spacetime will likewise satisfy the NEC. For pedagogical clarity, we will focus on the more illuminating $vv$-constraint here and relegate the remaining constraint equations to Appendix~\ref{sec:appendix}. Our analysis here will be twofold: we will first analyze the constraint \eqref{eq-deltaRvv} in the absence of $\delta g_{vv}$, separately compute the contribution of $\delta g_{vv}$, and sum the two together; this is possible so long as we work in the linearized regime.

By the Raychaudhuri equation, $R_{vv}$ depends only on the geometry of the $u=0$ hypersurface:
\be\label{eq-Rayray}
R_{vv} = -\partial_v \theta_{(v)} - {B_{(v)}}_{ij} {B_{(v)}}^{ij} - \kappa_{(v)} \theta_{(v)}.
\ee
Therefore in the absence of $\delta g_{vv}$ (i.e. implementing only the stretch):
\begin{align}
&R'_{vv}(u=0,v,y) = e^{-2s} R_{vv}(u=0,v e^{-s}, y)- (1-e^{-s})\theta_{(v)}(u=0,v=0,y) \delta(v) \nonumber\\
&= T'_{vv}(u=0,v,y)- (1-e^{-s})\theta_{(v)}(u=0,v=0,y) \delta(v)
\end{align}
where the delta function term results from the discontinuity in $\theta'_{(v)}$ across $v=0$. So $R'_{vv}-T'_{vv} = -(1-e^{-s})\theta_{(v)}[C] \delta(v)$. We now take the same perturbative limit of this transformation and re-introduce $\delta g_{vv}$:
\begin{align}\label{eq-vvinterm}
&-\frac{1}{2}\theta_{(u)} \partial_v \delta g_{vv} - \frac{1}{2} \nabla_{\perp}^2 \delta g_{vv}+ \chi^i \partial_i \delta g_{vv} +\nonumber\\
&\left(\nabla_{\perp}.\chi - \partial_v \theta_{(u)} - {B_{(v)}}_{ij} {B_{(u)}}^{ij} +8 \pi G (-T_{uv} -\mathcal{L}_{\text{matter}} + F_{uv}^2) \right)\delta g_{vv} \nonumber\\
&-\epsilon \theta_{(v)}[C] \delta (v) = 0
\end{align}
where all of the quantities multiplying $\delta g_{vv}$ and its derivatives are background quantities. We offer an alternative derivation of Eq. \eqref{eq-vvinterm} in Appendix \ref{sec:appendixvv} by implementing the ``stretch'' using an inaffinity shock~\cite{BouCha19, BouCha20}, which directly induces the delta function term in Eq. \eqref{eq-vvinterm}.

Since by construction, we are only perturbing the data on $\mathscr{H}^+_C$, $\delta g_{vv}(u=0, v=0^-, y)=0$, so the delta function term in Eq. \eqref{eq-vvinterm} enforces a jump in $\delta g_{vv}$:
\be\label{eq-gvvinit}
\delta g_{vv}(u=0,v=0^+,y) = 2\epsilon \frac{\theta_{(v)}[C]}{\theta_{(u)}[C]} \leq 0
\ee
where the sign comes from the fact that $\theta_{(v)}[C]\geq 0$ and $\theta_{(u)}[C]<0$.\footnote{In a non-generic spacetime, it is possible that $\theta_{(u)}[C]=0$, which raises potential concerns about the divergence in Eq. \eqref{eq-gvvinit}. We can avoid this issue by shifting the earliest location of the perturbation to some arbitrarily small $v=\delta>0$ instead of $v=0$. This new cut lies on a past causal horizon that reaches $\mathscr{I}$ and therefore cannot be stationary. In fact, originating the perturbation at $v=\delta>0$ makes sense physically since by sending $t_i$ to arbitrarily small values in the source Eq. \eqref{eq-source} we can affect the region arbitrarily close to $C$, but not $C$ itself. With this subtlety in mind, we pick $v=0$ as the origin of the perturbation in the main text because we can get arbitrarily close to $C$.} Note that by assumption $\theta_{(v)}[C]> 0$, and so $\delta g_vv < 0$, at least in a subset of $C$. This implies that the curve generated by $(\partial_v)^{a}$ is nowhere spacelike and at least timelike on a subset of $C$. In order to open up the lightcone and move the causal surface deeper into the bulk, it would be sufficient if $\delta g_{vv} \leq 0$ \textit{everywhere} on $\mathscr{H}^+_C$, not just at $C$.

We will now demonstrate that if $\delta g_{vv}(u=0,v=0^+,y)\leq 0$, then $\delta g_{vv}(u=0,v,y)\leq 0$ for all $v>0$, by analyzing the constraint that $\delta g_{vv}$ satisfies on $v>0$:
\begin{align}\label{eq-vvfinal}
-\frac{1}{2}\theta_{(u)} \partial_v \delta g_{vv} &- \frac{1}{2} \nabla_{\perp}^2 \delta g_{vv}+ \chi^i \partial_i \delta g_{vv}+ \nonumber \\
&\left(\nabla_{\perp}.\chi - \partial_v \theta_{(u)} - {B_{(v)}}_{ij} {B_{(u)}}^{ij} +8 \pi G (-T_{uv} -\mathcal{L}_{\text{matter}} + F_{uv}^2) \right) \delta g_{vv} = 0
\end{align}
which we may view as an ``evolution'' equation for $\delta g_{vv}$ on $u=0$ from which we can derive $\delta g_{vv}$ on $\mathscr{H}^+_C$ from its value at $C$.

It is not too difficult to see why $\delta g_{vv}(u=0,v,y)\leq 0$ for all $v>0$ starting from $\delta g_{vv}(u=0,v=0^+,y)\leq 0$. Suppose $\delta g_{vv}>0$ at some value of $v$. Then, assuming that all quantities in \eqref{eq-vvfinal} are continuous, there must exist a ``last'' constant-$v$ slice $\sigma$ on which $\delta g_{vv} \leq 0$ everywhere. By continuity, there exists a point $p\in \sigma$ where $\delta g_{vv}|_{p}=0$ (and then immediately becomes positive for larger $v$). By construction, we must have $\partial_i \delta g_{vv}|_{p} = 0$ and $\nabla_{\perp}^2 \delta g_{vv}|_{p} \leq 0$. But by \eqref{eq-vvfinal} this implies that $\partial_{v} \delta g_{vv}|_{p} \leq 0$, and so $\delta g_{vv}$ cannot become positive.

This reasoning may seem a bit fast, but it can be made more rigorous (and free of simplifying assumptions) using standard techniques. The operator $\mathscr{L}$ acting on $\delta g_{vv}$ in \eqref{eq-vvfinal} is parabolic whenever $\theta_{(u)}<0$; it thus satisfies the weak comparison principle for parabolic operators, which states that if $f$ and $h$ are functions satisfying $\mathscr{L}f\leq 0$ and $\mathscr{L}h\geq 0$ everywhere in the interior of the parabolic domain, and $f\leq h$ on the boundary of the parabolic domain, then $f\leq h$ everywhere in the parabolic domain. Setting $f=\delta g_{vv}$ and $h\equiv 0$, we immediately find that $\mathscr{L}f=0$ and $\mathscr{L}h=0$, so the weak comparison principle yields the desired conclusion: $\delta g_{vv}\leq 0$ everywhere on $\mathscr{H}^+_C$. Technically, the weak comparison principle is usually stated for domains in $\mathbb{R}^{n}$, fortunately, it follows as a fairly direct consequence of the maximum principle for elliptic operators, which does hold for more general manifolds~\cite{AndGal98}. The functions $f$ and $h$ need only be of Sobolev type $W^{1,2}_{0}$; that is, only their local weak derivatives are required to exist~\cite{GilbargTrudinger}, which is sufficient for our purposes. In fact, a version of the maximum principle for elliptical operators on ``rough'' null hypersurfaces (including caustics and non-local intersections specifically on event horizons) was proved in~\cite{Gal99}.

To make sure that our $\delta g_{vv}$ solution exists, we need to also satisfy the other constraints \eqref{eq-deltaRuv}, \eqref{eq-deltaRvi}, and \eqref{eq-deltaRij}. This is easy to do because they are ``evolution'' equations for $\partial_u \delta g_{ij}$ and $\partial_u \delta g_{vi}$ on $\mathscr{H}^+_C$ which we can solve no matter what $\delta g_{vv}$ is. We relegate this discussion to appendix \ref{sec:appendix}.

Let us now discuss possible subtleties in our construction due to caustics. Since caustic lines will generically be a measure zero subset of $\mathscr{H}^+$~\cite{Rademacher1919, HawEll, Per04}, we believe that they do not pose a fundamental obstruction to our procedure. Caustic lines can intersect $C$, at which point $C$ will generically be kinked. At the location of the kink, a chunk of transverse directions, associated to the generators that emanate from the caustic line, needs to get inserted in the transverse domain on which we place our boundary data for Eq. \eqref{eq-vvfinal}. However, so long as this data satisfies $\delta g_{vv} \leq 0$, Eq. \eqref{eq-vvfinal} still guarantee $\delta g_{vv} \leq 0$ everywhere on $\mathscr{H}^+_C$. In fact, since these new generators do not extend to the past of the caustics by definition, we expect to have even more freedom in specifying this boundary data because we do not have to worry about how this boundary is glued to some past hypersurface.

Let us offer an alternative argument to further ameliorate caustic-related worries. As $\mathscr{H}^+_C$ settles down to Kerr-Newman, there exists an earliest cross section $\mu_{\text{earliest}}$ lying on a past horizon with no caustics in its future. By setting $\mu_{\text{earliest}}$ as the origin of our perturbation (the new $v=0$), we can avoid caustics altogether. Furthermore, each perturbation should make the portion of the horizon to the future of $\mu_{\text{earliest}}$ more stationary, pushing the new $\mu_{\text{earliest}}$ further to the past eventually approaching $\mathscr{H}^-$.

Lastly, it is important to show that the perturbation has not shifted the $u=0$ surface to the point where it is no longer close to the event horizon -- especially in the asymptotic region $v \to \infty$. This is simplest to do if we assume that the \emph{background} horizon settles down to a stationary spherically symmetric configuration at some finite affine parameter, though we expect proper falloffs to hold more generally. The equation for $\delta g_{vv}$ then simplifies to:
\begin{align}\label{eq-deltagvvspherical}
-\frac{1}{2}\theta_{(u)} \partial_v \delta g_{vv} &- \frac{1}{2} \nabla_{\perp}^2 \delta g_{vv}-\partial_v \theta_{(u)}~\delta g_{vv} = 0
\end{align}
On $\mathscr{H}^+$, we have $\theta_{(u)}\sim v$ asymptotically. We can then solve for the asymptotic behavior of $\delta g_{vv}$ from Eq. \eqref{eq-deltagvvspherical}:
\be
\delta g_{vv} \sim v^{-2}
\ee


Therefore $u=0$ asymptotes to a stationary null hypersurface after the perturbation, so it naturally lines up with the new causal horizon in the $v\to+\infty$ limit. We find that our proposed perturbation results in a spacetime that solves the Einstein equation and in which the causal horizon is pushed deeper into the bulk unless $\theta_{(k)}[C]=0$.

\section{Zigzag} \label{sec:ZigZag}

We will now use the above perturbation to show that the causal surface $C$ can be moved arbitrarily close to the outermost extremal surface $X$, the appetizer of the lunch, using simple sources only. This requires us to show that (1) the perturbation analyzed in the previous section can be engineered from simple sources on the boundary, (2) the perturbation can be iteratively repeated both in the past and future, resulting in the approach of $C$ to $X$, without incurring high complexity, and (3) that this procedure does not change the geometry of the lunch, nor does the causal surface breach the lunch region in the process. We will begin our discussion by assuming for simplicity\footnote{No pun intended.} that the causal surface and the appetizer have the same topology; topological differences between the two surfaces are discussed at the end.

To show (1), we simply employ our assumption of the validity of HKLL, discussed in Sec.~\ref{sec:intro}, and evolve the data on $\mathscr{H}^+_C \cup \mathscr{H}^-_C$ ``sideways'' towards $\mathscr{I}$ to find appropriate boundary conditions at $\mathscr{I}$, which will be smeared local sources.\footnote{We are aware that some mathematician somewhere must be apoplectic with rage after reading this paragraph. Since we are physicists, we will nevertheless persevere under the usual assumption in holography that HKLL does in fact work out. Concerned readers who do not a priori wish to subscribe to HKLL may take comfort in the following interpretation of our results: we show that reconstruction of the simple wedge is no more complex than that of the causal wedge.} This sideways evolution was also used in~\cite{EngWal18} to prove the simple entropy conjecture in the case where the horizon was only perturbatively non-stationary.

We will shortly demonstrate (2) in detail. The process is similar in spirit to the zigzag process in Sec. \ref{sec:JT}, but instead of removing infalling chiral modes, we apply our perturbation in Sec. \ref{sec:perturbation}. First, we will discuss the consequences of repeated iterations of our perturbation on the future horizon and then add timefolds into our procedure.

Let us begin by a comparison between the perturbed and unperturbed causal surfaces, denoted $C'$ and $C$ respectively, as a result of one instance of our perturbation on $\mathscr{H}^+_C$. Since the perturbation is localized away from the past event horizon $\mathscr{H}^-$, it is expedient to compare the relative location of $C$ with that of $C'$ using their position on the past event horizon: both are slices of $\mathscr{H}^-$. Note that in the perturbed geometry no special role is played by $C$. As shown in Sec.~\ref{sec:perturbation}, the perturbation guarantees that $C'$ is ``inwards'' (i.e. at larger $u$) compared with $C$ so long as $\theta_{(k)}[C]\neq 0$. Nothing stops us from then repeating the perturbation above on the \emph{new} future horizon. As long as some point on the causal surface satisfies $\theta_{(k)}>0$, the perturbation pushes the causal wedge further inwards.

The only obstruction in the construction above occurs if $\theta_{(k)}$ vanishes identically on the causal surface. Thus it is clear that the inwards shift of the causal surface obtained via simple sources limits to an outermost marginally outer trapped surface $\mu$. (We can define a rigorous notion of the causal surface approaching arbitrarily close to a surface with $\theta_{(k)}=0$ by picking an affine parameterization on $\mathscr{H}^-$ and defining proximity of the two surfaces in terms of the maximal elapsed affine parameter between them.) 

Let us provide intuition for the existence of $\mu$. Say on $\mathscr{H}^-$ we can identify two cuts $\mu_1$ and $\mu_2$ such that $\mu_2$ encloses $\mu_1$ and $\theta_{(k)}[\mu_1]\leq 0$, $\theta_{(k)}[\mu_2] \geq0$. Then we expect an outermost marginally outer trapped surface $\mu$ in-between $\mu_1$ and $\mu_2$.\footnote{This is proved when $\mu_1$ and $\mu_2$ satisfy $\theta_{(k)}[\mu_1]<0$, $\theta_{(k)}[\mu_2] >0$ and lie on smooth spacelike slices in $D=4$ in~\cite{AndMet07}, but since we can approach $\mathscr{H}^-$ with such spacelike slices we expect the same to hold for $\mathscr{H}^-$, and the technology used in the extant proof is expected to generalize to higher dimensions.} On $\mathscr{H}^-$, $C$ plays the role of $\mu_2$. For $\mu_1$, we can pick $\partial J^{+}[X] \cap \mathscr{H}^{-}$ -- whenever it is a full cross section of $\mathscr{H}^{-}$ -- which satisfies $\theta_{(k)}\leq 0$ by the focusing theorem. However, note that $\partial J^{+}[X] \cap \mathscr{H}^{-}$ might be empty if $\mathscr{H}^{-}$ falls into a singularity before intersecting $\partial J^{+}[X]$. Even so, at least for Kasner-like singularities, we can find cross sections of $\mathscr{H}^-$ in a neighborhood of the singularity which are trapped~\cite{Wal12}. Note also that generically, our choices for $\mu_1$ and $\mu_2$ satisfy $\theta_{(k)}[\mu_1]<0$, $\theta_{(k)}[\mu_2] >0$. We will then have candidates for both $\mu_1$ and $\mu_2$, so $\mu$ exists.

Prima facie the procedure at this point appears to have failed! The causal surface will generically stop well away from null separation with $X$, and even further away from coincidence with $X$. However, this is only true on the particular timefold under consideration. To proceed to close the gap further, we reverse the arrow of time. We can then repeat the procedure above in time reverse to shrink the discrepancy between the causal wedge and $W_X$. We iterate this procedure via forward and reverse timefolds; each step brings the causal surface and $X$ closer. Just like for the JT gravity case in Section \ref{sec:JT}, the causal surface should limit to the outermost extremal surface after sufficiently many timefolds. Importantly, since the bulk physics involved is entirely classical, the number of timefolds required for the causal surface to approach the outermost extremal surface, within a given precision, should be independent of the Planck scale. This means that the complexity of the process cannot diverge in the classical limit $N \to \infty$. We therefore conclude that for $X$ and $C$ of identical topology, the simple wedge and the outermost extremal wedge coincide.

Finally, we address (3): throughout this construction, the geometry of the lunch is left unaltered: the perturbation is localized in the causal complement of the lunch, guaranteeing that the lunch remains undisturbed.

What about the case where the topologies are different? This could for example be the case in a time-symmetric and spherically symmetric null shell collapse in AdS where $C$ is a sphere and $X=\varnothing$ \cite{EngWal17a}. If such topology difference exists, we would expect it to present itself between $C$ and the outermost marginally trapped surface on $\mathscr{H}^-$ in a timefold (or several timefolds) of our procedure above. Furthermore, because each iteration of our $\mathscr{H}^+_C$ perturbation pushes the causal surface inwards on $\mathscr{H}^-$ a bit, at some point the jump from $C$ to $C'$ would have to involve a topology change. As there is nothing explicit in our construction that constrains the topology of $C'$ according to that of $C$, we do not see a fundamental obstruction against such topology changes arising from our iterative procedure. A rigorous treatment of such cases might be interesting but is left to future work.

\section{Simplicity Killed the Python} \label{sec:appetizer}
Having established that the simple wedge is in fact reconstructible using exclusively simple experiments, we now explore the implications of our results beyond the converse to the Python's lunch: what is the dual to the area of the outermost extremal surface? What is the field theory interpretation of our results, and in particular what is the ``simple state'' dual to the simple wedge?

\subsection{Time-Folded Simple Entropy}
As noted in the introduction, the simple entropy of~\cite{EngWal17b, EngWal18} is a coarse-graining over high complexity measurements conducted after a fixed boundary time $t_{\mathrm{bdy}}$ on a single timefold:

\be \label{eq:SimpleEnt}
S_{\mathrm{simple}}[t_{\mathrm{bdy}},\rho_{\mathrm{bdy}}]= \max_{\rho\in{\cal B}} S_{\mathrm{vN}}[\rho]
\ee
where $\rho_{\mathrm{bdy}}$ is the actual state of the CFT, $t_{\mathrm{bdy}}$ is a choice of boundary time slice, and ${\cal B}$ is the set of all CFT states (density matrices) that have the same one-point functions as $\rho_{\mathrm{bdy}}$ under any simple sources turned on after the time $t_{\mathrm{bdy}}$ (and with some very late time cutoff to avoid recurrences). That is, ${\cal B}$ consists of the set of CFT states $\rho$ such that
\be
\left \langle E \mathcal{O}E ^{\dagger}\right \rangle_{\rho_{\mathrm{bdy}}} = \left \langle E \mathcal{O}E ^{\dagger}\right \rangle_{\rho},\ee
for all possible $E$ defined as in Eq.~\ref{eq-source}. The simple entropy at a given boundary time is thus a coarse-graining over high complexity data that preserves all of the simple data to the future (or past) of that time.

With these restrictions to a particular subset of boundary time and a fixed timefold, the simple entropy was proposed as a dual to the outer entropy, which is a bulk-defined quantity that coarse-grains over the exterior of an apparent horizon (a  surface which is by definition always marginally trapped\footnote{A refinement of an apparent horizon that was named a ``minimar surface''. Apparent horizons are generically minimar surfaces.}). The outer entropy coarse-grains over all possible spacetimes that look identical outside of a given apparent horizon to find the spacetime with the largest HRT surface, and thus the largest von Neumann entropy in the CFT:
\be
S_{\mathrm{outer}}[\mu]=\max\limits_{{\cal X}} \frac{\mathrm{Area}[X]}{4 G\hbar} =\max\limits_{\rho\in{\cal H}} S_{\mathrm{vN}}[\rho]=\frac{\mathrm{Area}[\mu]}{4 G\hbar},
\ee
where ${\mathcal X}$ consists of the HRT surfaces of all classical holographic spacetimes containing $O_{W}[\mu]$, and ${\cal H}$ is the corresponding set of CFT states; the final equality is proved in~\cite{EngWal17b, EngWal18}. This is done by discarding the spacetime behind $\mu$, constructing a spacetime with an HRT surface $X_{\mu}$ whose area is identical to the area of $\mu$, and then CPT conjugating the spacetime around $X_{\mu}$. By construction, $O_{W}[\mu]$ is left unaltered.

The proposal that the simple and outer entropies are identical says that there is a particular definition of black hole entropy which is a consequence of coarse-graining over the highly complex physics that we expect describes the interior:
\be
S_{\mathrm{outer}}[\mu(t_{\mathrm{bdy}})] = S_{\mathrm{simple}}[t_{\mathrm{bdy}}]
\ee
where $t_{\mathrm{bdy}}=\partial J^{-}[\mu]\cap \mathscr{I}$. Our construction in Section~\ref{sec:perturbation} establishes this conjecture for apparent horizons: in a given timefold, it is possible to push the event horizon all the way up to $\mu(t_{\mathrm{bdy}})$ without accessing any high complexity data for $t>t_{\mathrm{bdy}}$. 

Our construction is of course more general, as it applies to timefolds. Extending the simple entropy proposal to include timefolds immediately yields the holographic dual to the area of the outermost extremal surface $X$:
\be
\frac{\mathrm{Area}[X]}{4G\hbar}=S_{\mathrm{outer}}[X] = S_{\mathrm{simple}}
\ee
where $S_{\mathrm{simple}}$ is obtained from $S_{\mathrm{simple}}[t_{\mathrm{bdy}}]$ by taking $t_{\mathrm{bdy}}\rightarrow -\infty$ and including arbitrary timefolds. The inclusion of timefolds removes the need for a reference apparent horizon, and the coarse-grained spacetime (in which the outermost extremal surface is in fact the HRT surface) is obtained by CPT-conjugating around the outermost extremal surface $X$; see Fig.~\ref{fig:coarsegraining}. Crucially, note that the coarse-graining procedure leaves the outermost extremal wedge untouched and coarse-grains only over the lunch. Standard entanglement wedge reconstruction via quantum error correction~\cite{AlmDon14, DonHar16, CotHay17} applies to reconstruction of the outermost extremal wedge, since in the coarse-grained spacetime obtained by CPT conjugation, the outer wedge of $X$ is exactly the entanglement wedge. Since we will argue below that the coarse-grained spacetime has a simple modular Hamiltonian, entanglement wedge reconstruction using, for example, modular flow as in \cite{FauLew17} should be much simpler when based on the coarse-grained state rather than the original state. This is consistent with the simplicity of reconstructing the outermost extremal wedge.

\subsection{The Simple State}
So far we have introduced two manipulations that can be done to our original spacetime. The first was the zigzag procedure, introduced in Section \ref{sec:ZigZag} which makes the causal wedge coincide with (or become arbitrarily close to) the outermost extremal wedge. We say that the resulting spacetime is `exposed' because everything in the simple wedge can be directly seen by the boundary.\footnote{Note that this definition is unrelated to our definition of an `exposed surface' in the proof of Proposition \ref{prop-X}.} The second is the coarse-graining procedure introduced above where we CPT conjugate about the outermost extremal surface and thereby create a state where the outermost extremal wedge coincides with the entanglement wedge. If we apply \emph{both} manipulations to the spacetime, we can produce a spacetime where all three wedges (approximately) coincide. This is illustrated in Fig.~\ref{fig:diamond}.

\begin{figure}
    \centering
    \includegraphics[width=0.8\textwidth]{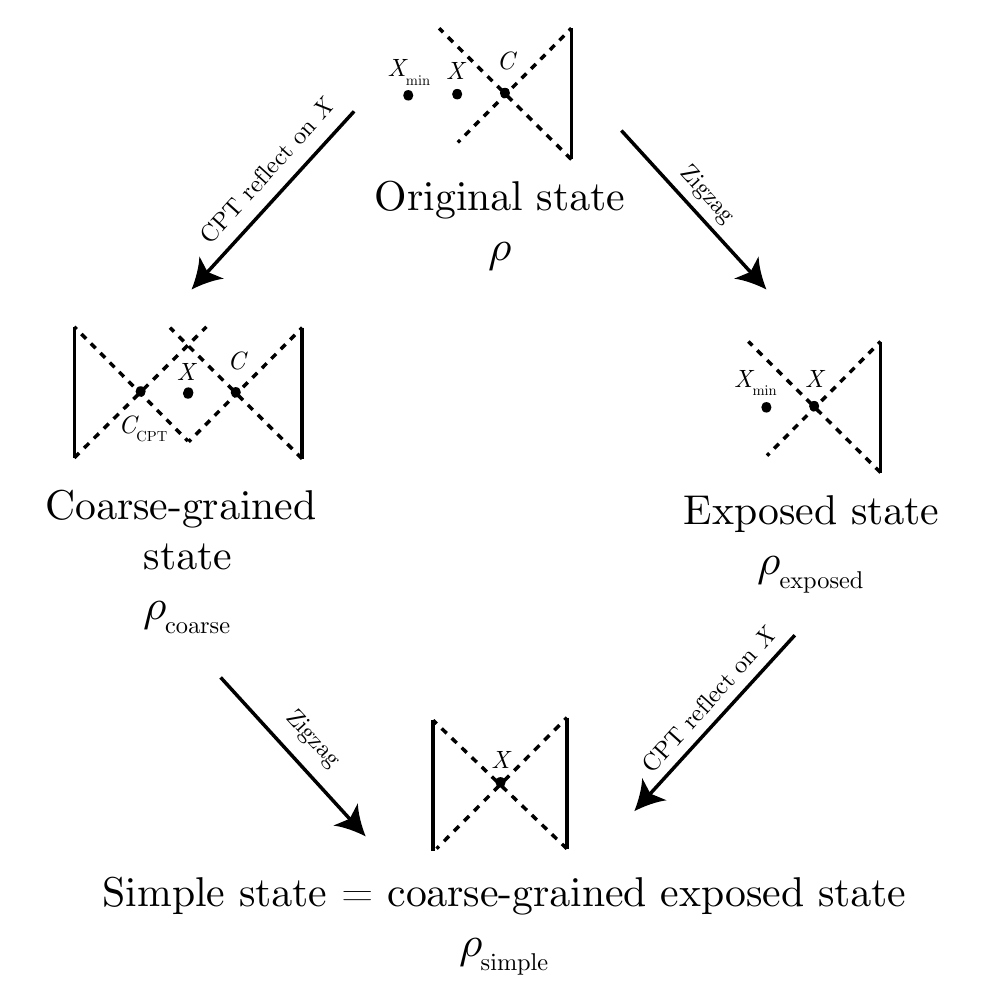}
    \caption{A diagram illustrating the relationships between the different states: the original CFT state $\rho$, which may be either coarse-grained to obtain $\rho_{\mathrm{coarse}}$ by forgetting about high complexity operators, or it may be ``exposed'' by acting on it with simple operators. The two operations modify causally independent and non-intersecting portions of the spacetime, so they commute: after obtaining the coarse-grained state in which $X$ is the HRT surface, we may perform our zigzag procedure to push the causal surface up to $X$ and obtain the simple state in which all three wedges coincide. Or, after obtaining the exposed state in which the causal and outermost extremal wedge coincide, but the entanglement wedge properly contains both, we may coarse-grain to obtain the same simple state.}
    \label{fig:diamond}
\end{figure}

That is, given any holographic CFT state $\rho$ dual to some entanglement wedge (which will likely have a Python's lunch), there exists a $\rho_{\mathrm{coarse}}$ which is indistinguishable from $\rho$ via simple experiments and has no Python's lunch. Executing our procedure zigzag on this coarse-grained state yields the coarse-grained exposed spacetime, in which the causal, simple, and entanglement wedges all coincide or come arbitrarily close to coinciding. The dual to this is described by the state obtained via the zigzag procedure together with the set of simple operators that remains from the final timefold. We shall refer to this history that includes the the remaining simple operators the simple history, and in a slight abuse of notation we will denote this entire history of states as $\rho_{\mathrm{simple}}$. 


What is the CFT interpretation of $\rho_{\mathrm{simple}}$? Below we prove that the causal and entanglement wedges coincide exactly if and only if the state dual to the entanglement wedge has a geometric modular flow. The immediate implication is that in the case where the zigzag procedure gives an exact coincidence, the simple state has an exactly local modular Hamiltonian, and rather than being a history it is in fact a single state. Note that this is suggestive of a CFT dual to a gravitational no-hair theorem. To be more precise: this result suggests that the set of stationary holographic black holes is to be identified with the high limited set of states with exactly local modular flow. If coincidence between the causal surface and the appetizer is asymptotic rather than exact, then we come to the conclusion that the modular flow generated by the simple state is very close to local in the sense that only operators with support in the asymptotically shrinking region between the causal and entanglement wedges are able to definitively tell that the two are not identical. Since that region translates (for simple operators) to access to arbitrarily late or early times, we find that finite-time simple measurements are unable to tell that the modular flow generated by $\rho_{\mathrm{simple}}$ at each stage is not local. 

The secondary implication is that it is possible to take any holographic state $\rho$ and, via a series of simple operations, render its modular flow (nearly) indistinguishable from a geometric flow via any simple experiments. If the appetizer has sufficient symmetry, then the statement should be true exactly. Let us now prove our theorem, which we do in broad generality for boundary subregions.

\begin{thm} Let $\mathscr{W}_{C}[R]$ denote the causal wedge of a boundary subregion $R$, and let $W_{E}[R]$ denote the entanglement wedge of the boundary subregion. $\mathscr{W}_{C}[R]=W_{E}[R]$ if and only if the boundary modular Hamiltonian on $R$ generates a geometric flow with respect to a boundary Killing vector field on $R$. 
\end{thm}

\begin{proof}
Assume that the boundary modular Hamiltonian generates a geometric flow with respect to some Killing vector field $\xi^{I}$ on $\partial M$ (here $I$ is a boundary spacetime index). Under modular flow, a local operator is mapped to another local operator:
\be
\mathcal{O}(x,s)= \rho_{R}^{-is/2\pi}\mathcal{O}(x)\rho_{R}^{is/2\pi}=\mathcal{O}(x_{\xi}(s))
\ee
where $x_{\xi}(s)$ is the boost along $\xi^{I}$ of $x\in D[R]$. Via~\cite{FauLew17}, $W_{E}[R]$ can be reconstructed by smearing the modular flow of local operators over $D[R]$. Since in this case modular flow is an automorphism on the space of local operators in $D[R]$, we can reconstruct all operators in $W_{E}[R]$ by smearing local operators:
\be
\Phi(x) = \int\limits_{D[R]} f(X|x)\mathcal{O}(X)dx,
\ee
where $f(X|x)$ is a smearing function. If there is a gap between $W_{E}[R]$ and $\mathscr{W}_{C}[R]$, then operators that are localized to the gap should commute with all local operators on the boundary (within our code subspace) via the extrapolate dictionary (note that this only works in the large-$N$ limit where we don't have to worry about gravitational dressing). However this is inconsistent with the equation above; so $x\in \mathscr{W}_{C}[R]$. But this argument holds for all local operators: there exist no local operators in the gap between $W_{X}[R]$ and $\mathscr{W}_{C}[R]$. In the large-$N$ limit (without backreaction), this means that there simply is no gap between the two wedges: $W_{X}[R]=\mathscr{W}_{C}[R]$.

To prove the other direction, we consider the proof of~\cite{FauLew17} for the zero-mode formula of entanglement wedge reconstruction (appendix B.1 of~\cite{FauLew17}). Starting with equation B.71, it is shown that the nonlocality of modular flow on $\partial W_{E}[R]$ is due to the change in the instrinsic metric of spatial slices of $\partial W_{E}[R]$ (in particular, the loss of ultralocality). When $\partial W_{E}[R]$ is stationary, the metric on codimension-two slices does not change with evolution along the congruence. This  means that the modular flow on $\partial W_{E}[R]$ is local, which in turn implies that the boundary modular flow is local as well. 
\end{proof}

\section{Discussion}\label{sec:disc}

Our primary technical result in this article is the proof of the converse to the Python's lunch proposal in the strict large-$N$ limit: operators that lie outside of a Python's lunch are simply reconstructible in the dual CFT, and moreover this reconstruction only relies on the bulk dynamics in the large-$N$ limit, manifestly respecting the causal structure of the background metric.

We emphasize that bulk reconstructions that are causal in this sense cannot work for the interior of the Python's lunch because no causal horizon can intersect the lunch. The CFT encoding of the Python's lunch appears to involve highly non-local quantum gravity effects. An example of such non-local dynamics is the ER=EPR conjecture~\cite{Van13, MalSus13} which asserts that the entanglement between an evaporating black hole and its Hawking radiation after the Page time must allow complicated operations on the distant radiation to change the state behind the lunch, drastically violating the naive causal structure dictated by the background metric. It has been speculated that wormhole-like ``corrections'' to the background geometry connecting the radiation to the black hole interior could explain the ``true'' causal structure not captured by the background metric. It is natural to speculate that similar dynamics are at play in the Python's lunch encoding into the boundary.

We will now discuss various generalizations of our main results:

\subsection{Boundary-Anchored Surfaces} \label{sec:bdyanchored}
We have focused here primarily on compact surfaces, but we may pose similar questions for boundary subregions: given the state $\rho_{R}$ on a boundary subregion $R$, how complex is the reconstruction of operators behind the event horizon but within the outermost extremal wedge? This requires a treatment of surfaces with a boundary-anchored component rather than surfaces whose components are all compact. Most of our results generalize almost immediately to the boundary-anchored case: Lemma 3~\cite{BouSha21} makes no reference the topology of surfaces (beyond the homology constraint); similarly for the proofs of Lemmas 1 and 2. The perturbed initial data prescription at the causal surface also carries over mutatis mutandis. As noted in Section~\ref{sec:perturbation}, the weak comparison principle operates on the basis of the maximum principle for elliptic operators. The latter does indeed apply to bounded domains in general, and to boundary-anchored hypersurfaces in AdS particular (see~\cite{EngFis19} for a discussion in the context of AdS/CFT). The main potential source of difficulty is the falloff: both the causal surface and the outermost extremal surface approach the asymptotic boundary, but not with the same tangent space; the asymptotic falloff of $\delta g$ must approach zero sufficiently fast so as to not spoil the asymptotics while bridging the gap between the two surfaces. We expect that the appropriate falloff conditions can be satisfied, but this remains a subject for future work.

\subsection{Robustness under Quantum Corrections}
Fundamentally, the classical calculations done in this paper are only interesting as an approximation to the fully quantum dynamics that actually describe the bulk in AdS/CFT. Do our arguments extend to the semiclassical setting where the background spacetime is still treated classically, but with quantum fields propagating on it? Do they generalize to the  regime where  perturbative corrections to the geometry, suppressed by powers $G\hbar$, are allowed to contribute? A number of important assumptions break down in this case: Raychaudhuri's equation is still valid, but the null energy condition will not generally hold, and so light rays emanating from a classical extremal surface can defocus. Fortunately, the quantum focusing conjecture (QFC) states that the generalized entropy of null congruences emanating from QESs is always subject to focusing. In a semiclassical or perturbatively quantum bulk, the appetizer is the outermost quantum, not classical, extremal surface; the QFC ensures that the outermost quantum extremal wedge always contains the causal wedge.\footnote{Apparent counterexamples are only possible when the time evolution includes interaction with an auxiliary system; once you include these auxiliary interacting modes when defining quantum extremality then the apparent causality issues go away.} The question is whether we can still expand the causal wedge using appropriate sources and timefolds in order to bridge the gap between the causal surface and the appetizer.

This is a much harder question than the classical question discussed here: the class of allowed QFT states is simply much harder to classify and use than classical field theory states. However, in particularly simple examples, for instance where the causal wedge is approximately Rindler-like, a quantum version of the ``left stretch'' appears to be well defined and gives exactly the right change in energy to reduce focusing and remove the perturbatively small, i.e. $O(1/N)$, distance between the causal and outermost extremal wedges~\cite{Levine:2020upy}. It is feasible that in the limit of many zigzags, the causal wedge would approach a Rindler-like region, at which point it becomes possible to apply the bulk unitaries discussed in \cite{Levine:2020upy} to eliminate the remaining small gap. We leave a detailed study of this to future work.

\subsection{Asymptotically Flat Spacetimes}
In asymptotically flat spacetimes, asymptotic infinity is lightlike rather than timelike. However, there do not seem to be any major obstructions to adapting the results of this paper to that setting. Instead of the timefolds of the asymptotically AdS problem, in asymptotically flat space one would presumably evolve forwards along future null infinity, then backwards with different boundary conditions that remove focusing at the past event horizon, in order to produce a state where the causal surface is very close to a past apparent horizon. Then one would evolve backwards and forwards along past null infinity in order to produce a state where the causal surface is very close to a future apparent horizon. At each step, the causal wedge increases in size. After sufficiently many such timefolds, the causal surface should approach the outermost extremal surface, as in the asymptotically AdS case. The interpretation of our results in the asymptotically flat case is naturally obfuscated by relatively inchoate status of flat holography. We may speculate that extremal surfaces are important more generally for defining entropy in gravity; it is also possible that a similar notion of a Python's lunch applies beyond AdS holography. We do not subscribe to any particular interpretation -- here we simply note that the technical aspects of this work are likely not restricted to AdS.

\subsection{Complexity Censorship}

Let us finish with a few comments on cosmic censorship, a prima facie unrelated conjecture about classical General Relativity. (Weak) cosmic censorship~\cite{Pen69} is essentially the statement that high curvature physics lies behind event horizons. One of its landmark consequences is that trapped surfaces lie behind event horizons, and that consequently marginally trapped and in particular extremal surfaces lie on or behind event horizons~\cite{HawEll}. It is clear from the above discussion that any violation of cosmic censorship would be quite problematic for the Python's lunch picture: if the nonminimal extremal surface could lie outside of the event horizon (or, in the quantum case, could communicate with $\mathscr{I}$), then operators behind would lie properly within the causal wedge and would thus be reconstructible by HKLL despite being exponentially complex. The Python's lunch proposal thus appears to depend heavily on the validity of cosmic censorship -- which is known to be violated in AdS~\cite{HorSan16, CriSan16, HorSan19}. As matters currently stand, violations of cosmic censorship notwithstanding, it is possible to prove that the holographic entanglement entropy prescription guarantees that trapped surfaces must lie behind event horizons~\cite{EngFol20}. We could however have proven the same statement from holographic complexity: marginally trapped surfaces (and therefore, trapped surfaces also, by the reasoning of~\cite{EngFol20}) must lie behind event horizons, for if they did not, operators behind the Python's lunch could be reconstructed in a simple procedure. This suggests that in AdS/CFT, aspects of cosmic censorship may be reformulated as ``complexity censorship'': that high complexity physics must be causally hidden and thus unable to causally communicate to $\mathscr{I}$.

\section*{Acknowledgments}
It is a pleasure to thank S. Alexakis, R. Bousso, S. Fischetti, L. Susskind for helpful discussions. NE is supported by NSF grant no. PHY-2011905, by the U.S. Department of Energy under grant no. DE-SC0012567 (High Energy Theory research), and by funds from the MIT physics department. GP is supported by the UC Berkeley physics department, the Simons Foundation through the "It from Qubit" program, the Department of Energy via the GeoFlow consortium and also acknowledges support from a J. Robert Oppenheimer Visiting Professorship. ASM is supported by the National Science Foundation under Award Number 2014215.

\appendix

\section{uv and vi constraints} \label{sec:appendix}

Here we write down the general structure of the perturbative constraints \eqref{eq-deltaRuv}, \eqref{eq-deltaRvi}, \eqref{eq-deltaRij}:
\bea
 \frac{1}{2}g^{ij} \partial_v \partial_u \delta g_{ij} = (\hat{L}_1)^{ab} \partial_u \delta g_{ab}+(\hat{L}_2)^{ab} \delta g_{ab} + 8\pi G \delta T_{u v} \label{eq-uv}\\
\partial_v \partial_u \delta g_{vi} = (\hat{L}_1')^{ab}_i \partial_u \delta g_{ab}+(\hat{L}_2')^{ab}_i \delta g_{ab} + 8\pi G \delta T_{v i}\label{eq-vi}\\
 \partial_v \partial_u \hat{\delta g_{ij}} = (\hat{L}_1'')^{ab}_{ij} \partial_u \delta g_{ab}+(\hat{L}_2'')^{ab}_{ij} \delta g_{ab}+ 8 \pi G \hat{\delta T_{ij}} \label{eq-ij}
\eea
where $\hat{L_n}$, $\hat{L'_n}$, and $\hat{L''_n}$ represent linear differential operators that depend only on the background metric and do not have $u$ derivatives. The only significance of Eqs. \eqref{eq-uv}, \eqref{eq-vi}, and \eqref{eq-ij} for us is that they are consistent with the $\delta g_{vv}$ solution to Eq. \eqref{eq-vvfinal}. Plugging in the $\delta g_{vv}$ solution reduces Eqs. \eqref{eq-uv}, \eqref{eq-vi}, and \eqref{eq-ij} to coupled linear differential equations involving $\partial_u \delta g_{ij}\rvert_{u=0}$ and $\partial_u \delta g_{vi}\rvert_{u=0}$ which are first order in $\partial_v$. Note that the initial conditions at $v=0$ are $\partial_u g_{ij}(u=0,v=0^-,y) = \partial_u g_{vi}(u=0,v=0^-,y)=0$, as enforced by the perturbation being localized away from $\mathscr{H}^-$.

\section{Alternative derivation of the vv constraint} \label{sec:appendixvv}

Here we provide an alternative derivation for the same perturbation discussed in Sec. \ref{sec:perturbation}. Instead of the transformations \eqref{eq-gijboost}, \eqref{eq-mattervvboost}, and \eqref{eq-kappaboost} on $\mathscr{H}^+_C$, we can equivalently insert an inaffinity shock at $v=0$~
\cite{BouCha19, BouCha20}:
\begin{align}\label{eq-kappashock}
    \kappa_{(v)} = (1-e^{-s}) \delta(v)
\end{align}
and take the $(1-e^{-s}) \sim \epsilon$ limit. In addition, we want to introduce the following $\delta g_{ab}$ transformation:
\begin{align}\label{eq-linearpert}
    ds^2 = -2 du dv + (g_{vv}+\delta g_{vv}) dv^2 + 2(g_{vi}+\delta g_{vi}) dv dy^i + (g_{ij} + \delta g_{ij}) dy^i dy^j,
\end{align}
with
\begin{align}
    \delta g_{vi}\rvert_{u=0}&=0\\
    \delta g_{ij}\rvert_{u=0}&=0\\
    \partial_u \delta g_{vv}\rvert_{u=0}&=0
\end{align}
We need to apply the combined perturbations in Eqs. \eqref{eq-kappashock} and \eqref{eq-linearpert} to the $vv$ constraint:
\begin{align}
    \delta G_{vv} + \Lambda \delta g_{vv} = 8\pi G \delta T_{vv}
\end{align}
The only contribution in $\delta G_{vv}$ from Eq. \eqref{eq-kappashock} is through the $\kappa_{(v)} \theta_{(v)}$ term in Eq. \eqref{eq-Rayray}. Summing this up with the contribution from the contribution from $\delta g_{ab}$ of Eq. \eqref{eq-linearpert}, we get:
\begin{align}
&-\frac{1}{2}\theta_{(u)} \partial_v \delta g_{vv} - \frac{1}{2} \nabla_{\perp}^2 \delta g_{vv}+ \chi^i \partial_i \delta g_{vv} +\nonumber\\
&\left(\nabla_{\perp}.\chi - \partial_v \theta_{(u)} - {B_{(v)}}_{ij} {B_{(u)}}^{ij} +8 \pi G (-T_{uv} -\mathcal{L}_{\text{matter}} + F_{uv}^2) \right)\delta g_{vv} \nonumber\\
&-\epsilon \theta_{(v)}[C] \delta (v) = 0
\end{align}
Similarly, the uv and vi constraints could be analyzed resulting in Eqs. \eqref{eq-uv} and \eqref{eq-vi}.

\bibliographystyle{jhep}
\bibliography{all}

\end{document}